\documentclass{file}
\usepackage{pat}
\usepackage{paralist}
\usepackage{dsfont}  
\usepackage[utf8x]{inputenc}
\usepackage{skull}
\usepackage{paralist}

\newtheorem{iterationelltemp}{Iteration}
\newenvironment{iterationell}
{%

  \begin{iterationelltemp}%
}
{%
  \end{iterationelltemp}%
}
\usepackage{centernot}
\usepackage{mathtools}
\usepackage{stmaryrd}
\usepackage{algorithm,algpseudocode}
\usepackage{tcolorbox}
\usepackage{graphicx}
\usepackage{bbm}  
\usepackage{listings}
\usepackage{pifont}%
\usepackage[normalem]{ulem}
\usepackage{cancel}
\usepackage{thmtools}

\usepackage{todonotes}

\usepackage{etoolbox}

\usepackage{float}
\usepackage{tabularx}
\usepackage{array}
\usepackage{subfig}
\usepackage{graphicx}
\usepackage{multirow}
\usepackage{makecell}

\newtheorem{observation}{Observation}
\newtheorem{definition}{Definition}

\newtheorem{procedure}{Procedure}
\newtheorem{lemma}{Lemma}
\newtheorem{iter}{Iteration}

\newtheorem{proposition}{Proposition}

\newtheorem{notation}{Notation}
\newtheorem{example}{Example}
\newtheorem{claim}{Claim}
\newtheorem{corollary}{Corollary}

\newtheorem{remark}{Remark}
\definecolor{maroon}{rgb}{0.5, 0.0, 0.0}
\definecolor{darkblue}{rgb}{0.0, 0.0, 0.55}

\newcommand{\np}{\mathbb{NP}}
\newcommand{\p}{\mathbb{P}}
\newcommand{\greduction}{\mathbb{G}}
\newcommand{\greductiontwo}{\mathbb{G}}
\newcommand{\setofremovededges}{\mathbb{C}}
\newcommand{\removedcrossedge}{e_{c}}
\newcommand{\removedcrossedgeendpointone}{u_{c}}
\newcommand{\removedcrossedgeendpointtwo}{v_{c}}
\newcommand{\removedcrossedgeendpoints}{(u_{c}, v_{c})}
\newcommand{\edgesbetweentitj}{E'_{i,j}}
\newcommand{\edgesbetweentitjprime}{E''_{i,j}}
\newcommand{\unsaturatededge}{e}
\newcommand{\unsaturatededgeendpoints}{(u,v)}
\newcommand{\addededge}{e'}
\newcommand{\vreductiontwo}{\mathbb{V}}
\newcommand{\ereductiontwo}{\mathbb{E}}
\newcommand{\treduction}{\mathbb{T}}
\newcommand{\numitems}{|\mathcal X|}
\newcommand{\numsets}{|\mathcal F|}
\newcommand{\greductionsj}{\mathbb{G}_j}
\newcommand{\bigoh}{\mathcal{O}}%
\usepackage[longnamesfirst,numbers,sort&compress]{natbib}

\usepackage[mathlines]{lineno}
\newcommand*\patchAmsMathEnvironmentForLineno[1]{%
 \expandafter\let\csname old#1\expandafter\endcsname\csname #1\endcsname
 \expandafter\let\csname oldend#1\expandafter\endcsname\csname end#1\endcsname
 \renewenvironment{#1}%
    {\linenomath\csname old#1\endcsname}%
    {\csname oldend#1\endcsname\endlinenomath}}%
\newcommand*\patchBothAmsMathEnvironmentsForLineno[1]{%
 \patchAmsMathEnvironmentForLineno{#1}%
 \patchAmsMathEnvironmentForLineno{#1*}}%
\AtBeginDocument{%
\patchBothAmsMathEnvironmentsForLineno{equation}%
\patchBothAmsMathEnvironmentsForLineno{align}%
\patchBothAmsMathEnvironmentsForLineno{flalign}%
\patchBothAmsMathEnvironmentsForLineno{alignat}%
\patchBothAmsMathEnvironmentsForLineno{gather}%
\patchBothAmsMathEnvironmentsForLineno{multline}%
}

\newcommand{\iteration}{I}
\newcommand{\alg}{\operatorname{ALG}}
\newcommand{\opt}{\operatorname{OPT}}

\definecolor{brightmaroon}{rgb}{0.76, 0.13, 0.28}
\definecolor{linkblue}{rgb}{0, 0.337, 0.227}
\newrobustcmd{\onesub}{\mathord{\includegraphics{figs/one-sub}}}
\newrobustcmd{\leftup}{\mathord{\includegraphics{figs/left-up}}}

\makeatletter
\newcommand{\xMapsto}[2][]{\ext@arrow 0599{\Mapstofill@}{#1}{#2}}
\def\Mapstofill@{\arrowfill@{\Mapstochar\Relbar}\Relbar\Rightarrow}
\makeatother

\title{{Algorithms and Hardness Results for the $(k,\ell)$-Cover Problem\thanks{Preliminary version accepted to CALDAM 2025}}}
\author{
Amirali Madani\thanks{School of Computer Science, Carleton University, Ottawa, Ontario, Canada. Supported by the Natural Sciences and Engineering Research Council of Canada (NSERC). Emails: {\tt amiralimadani@cmail.carleton.ca}, {\tt anil@scs.carleton.ca}, and {\tt bobby.miraftab@gmail.com}}
\;\;\;\;\;\;Anil Maheshwari\footnotemark[2]\;\;\;\;\;\;Babak Miraftab\footnotemark[2]\;\;\;\;\;\;
Bodhayan Roy\thanks{Department of Mathematics, Indian Institute of Technology Kharagpur, India. Supported by the Science and Engineering Research Board (SERB) via the project MTR/2021/000474. Email: {\tt bodhayan.roy@gmail.com}.
}
}

\date{}

\begin{document}

\maketitle

\begin{abstract}
A connected graph has a $(k,\ell)$-cover if each of its edges is contained in at least $\ell$ cliques of order $k$. Motivated by recent advances in extremal combinatorics and the literature on edge modification problems, we study the algorithmic version of the $(k,\ell)$-cover problem. Given a connected graph $G$, the $(k, \ell)$-cover problem is to identify the smallest subset of non-edges of $G$ such that their addition to $G$ results in a graph with a $(k, \ell)$-cover. For every constant $k\geq3$, we show that the $(k,1)$-cover problem is $\mathbb{NP}$-complete for general graphs.
Moreover, we show that for every constant $k\geq 3$, the $(k,1)$-cover problem admits no polynomial-time constant-factor approximation algorithm unless $\mathbb{P}=\mathbb{NP}$. However, we show that the $(3,1)$-cover problem can be solved in polynomial time when the input graph is chordal. For the class of trees and general values of $k$, we show that the $(k,1)$-cover problem is $\mathbb{NP}$-hard even for spiders. However, we show that for every $k\geq4$, the $(3,k-2)$-cover and the $(k,1)$-cover problems are constant-factor approximable when the input graph is a tree. \\

\end{abstract}

\section{Introduction}
In recent years, research on \defin{edge modification problems} has gained a lot of attention. The area of edge modification problems spans many definitions. Still, many such problems ask for the optimal way of editing an input graph $G$ to another graph $G'$ with a desired property, usually through edge additions to $G$~\cite{edgemodificationsurvey}. The minimization objective is often defined in the literature as the number of added edges. Edge modification problems have been studied for many graph properties. For instance, some works have studied editing graphs into being Eulerian, regular, or having a specific degree sequence through edge additions and removals~\cite{editing2,editing3,editing4}. For a more comprehensive survey of edge modification problems, see \cite{edgemodificationsurvey}.

This paper is specifically motivated by the intersection of edge modification problems and the community search problem, the latter of which has numerous applications in data science. Extensive research has been conducted on transforming input graphs into \defin{cluster graphs} through edge additions and removals~\cite{cluster1, cluster2, cluster3, cluster4}. A cluster graph is a graph whose every connected component is a complete graph (a clique). These graphs have various applications, such as modeling different communities and grouping similar items within a network into the same cluster~\cite{clusterapplication}.
In community search, cohesive subgraphs are commonly used to model communities. Many studies define the cohesiveness of a subgraph by its minimum degree~\cite{corekeywordcommunitysearch}, while others consider a subgraph cohesive if each of its edges is covered by cliques of order three~\cite{trusskeywordcommunitysearch}. For the former measure of cohesion, Fomin, Sagunov, and Simonov~\cite{fomin2023building}, along with Chitnis and Talmon~\cite{chitnis}, have studied relevant edge modification problems. These problems aim to construct large subgraphs with a specified minimum degree by adding a few edges.

Graphs with local covering conditions on edges or vertices have also attracted significant interest in extremal graph theory. Burkhardt, Faber, and Harris~\cite{BurkhardtFH20} established asymptotically tight lower bounds on the number of edges in connected graphs where every edge lies in at least $\ell$ triangles. Chakraborti and Loh~\cite{chakraborti2020extremal} provided tight lower bounds, along with characterizations of extremal graphs, on the number of edges in graphs where every vertex belongs to a clique of order \( k \geq 3 \). Motivated by these results and the many applications of graphs with local edge covering conditions in big graph-based data analysis, Chakraborti et al.~\cite{chakraborti2024sparse} introduced the concept of \((k, \ell)\)-covers. A connected graph \( G \) has a \defin{\((k, \ell)\)-cover} if every edge of \( G \) lies in at least \( \ell \) copies of \( K_k \) (a clique of order \( k \)). They proved tight lower bounds and structural characterizations of graphs with \((k, 1)\)-covers (\( k \geq 3 \)) and graphs with \((3, 2)\)-covers. Motivated by all of this, we study the algorithmic version of the \((k, \ell)\)-cover problem.
\subsection{Preliminaries}
We present some definitions. \begin{definition}
    Let $G=(V,E)$ be a graph and let $E' \subseteq (V \times V) \sm E$ be a set of non-edges of $G$. $E'$ is a \defin{$(k,\ell)$-completion set} of $G$ if $G \cup E'$ has a $(k,\ell)$-cover.
\end{definition}
For simplicity, we refer to a $(k,1)$-completion set as a \defin{$k$-completion set}.
\begin{definition}
    The $(k,\ell)$-cover problem: Given a connected graph $G=(V,E)$, two integers $k\geq 3$ and $t\geq 0$, does $G$ have a $(k,\ell)$-completion set of size at most $t$?
\end{definition}
A graph $G$ is chordal if it does not have the cycle of length at least four as an induced subgraph. Unless specified otherwise, $G = (V, E)$ always serves as a connected graph with $ n \geq 3$ vertices throughout this paper. For an edge $e\in E$, we say $e$ is \defin{$k$-unsaturated} if it is not contained in any cliques of order $k$ in $G$ for $k\geq 3$. 

We call a graph $G$ \defin{non-trivial} if it has at least two vertices. For a graph $G$ and subset $S\subseteq V(G)$ of its vertices, we denote the subgraph of $G$ induced on $S$ by $G[S]$. For $G$ and $v \notin V$, we denote the operation $V \cup \{v\}$ by $G \cup v$. Similarly, the operation $G\cup e$ for an edge $e=(u,v)$ results in a graph $G'=(V',E')$ with $V'=V\cup \{u,v\}$ and $E'=E\cup\{(u,v)\}$. Within our algorithms, we sometimes initialize an empty graph as $H\xleftarrow{} \emptyset$. This operation constructs a graph $H=(V,E)$ with $V=E=\emptyset$. For further graph-theoretic and algorithmic notations and definitions not defined in the paper, we refer the reader to~\cite{bondy1982graph} and~\cite{cormen}. 
\subsection{New Results}
Our first set of hardness results (\Cref{nphardnesskeq3}, \Cref{nphardnesskgeq}, and \Cref{thm1}) state that for every constant $k\geq 3$, the $(k,1)$-cover problem is $\np$-complete for general graphs and admits no constant-factor approximation algorithm running in polynomial time unless $\p=\np$. However, we show that the $(3,1)$-cover problem can be solved in polynomial time when its input graph is restricted to the class of chordal graphs (\Cref{thmchordal}). For general values of $k$, we show in \Cref{thmhardnesstree} that the $(k,1)$-cover problem remains $\np$-hard on trees, even when the tree is restricted to the class of spiders. However, we show that for every $k\geq 4$, the $(k,1)$-cover and the $(3,k-2)$-cover problems are constant-factor approximable for trees (\Cref{thmtwoapprox} and \Cref{thmgivefour}).  

The remainder of this paper is organized as follows. \Cref{sechardnesss} contains our hardness results for general graphs and trees. In \Cref{sechardnessgeneral}, we study the hardness of the $(k,1)$-cover problem on general graphs for $k=3$ (\Cref{reductionk3}) and $k\geq 4$ (\cref{reductionkgeq4}). In \Cref{hardnessapprox}, we show the hardness of approximation for the $(k,1)$-cover problem on general graphs. We conclude \Cref{sechardnesss} in \Cref{nphardnesstrees} by proving the $\np$-hardness of the $(k,1)$-cover problem for spiders and general values of $k$.  In \Cref{secchordal}, we first present an optimal algorithm for the $(3,1)$-cover problem for the class of trees (\Cref{sectrees}, \Cref{proptreeoptimal}). We use the algorithm for trees to present an optimal algorithm for the class of chordal graphs in \Cref{secchordalmain}. Our approximation algorithms for trees are described in \cref{sec4}. Finally, we conclude the paper in \Cref{conclude} by presenting some potential avenues for future research.

\section{Hardness Results}
In this section, we show some hardness results for the $(k,1)$-cover problem for general graphs (\Cref{sechardnessgeneral}) and the class of trees (\Cref{nphardnesstrees}).
\label{sechardnesss}
\subsection{Hardness Results for General Graphs}
\label{sechardnessgeneral}
In this section, we show that for every constant $k\geq 3$, a connected graph $G=(V,E)$, and an integer $t$, it is $\np$-complete to decide whether there exists a set $S \subseteq (V \times V) \setminus E$ with $|S|\leq t$ such that $G \cup S$ has a \mbox{$(k,1)$-cover}. Moreover, we show these problems are also hard to approximate within a constant factor. 

We reduce the well-known problem of SET-COVER to the $(k,1)$-cover problem for every $k\geq3$. The decision problem of SET-COVER is formally stated as follows. An instance $(\mathcal{X}, \mathcal{F},t)$ of SET-COVER consists of a finite set $\mathcal{X}$ of items, a family of subsets $\mathcal{F}$ of $\mathcal{X}$ such that no set in $\mathcal{F}$ is empty and every item in $\mathcal{X}$ belongs to at least one set from $\mathcal{F}$, and an integer $t$. Given an instance of SET-COVER, the problem is to determine whether there exists a subset $\mathcal{S} \subseteq \mathcal{F}$ with $|\mathcal{S}|\leq t$ such that the sets in $\mathcal{S}$ cover all items of $\mathcal{F}$, i.e., $\bigcup_{S \in \mathcal{S}} S=\mathcal{X}$. We say that a subset $S \in \mathcal{F}$ \defin{covers} its items, and each item $x_i \in \mathcal{X} \cap S$ is \defin{covered} by $S$. It is well known that SET-COVER is $\np$-complete \cite{karp1975computational}. We provide two separate $\np$-completeness proofs for $k=3$ and $k\geq 4$, where both reductions are from SET-COVER. We do so because the second reduction can be generalized to any $k\geq 4$, but the reduction graph already has a $(3,1)$-cover.

The remainder of this section is organized as follows. In \Cref{reductionk3} and \Cref{reductionkgeq4}, we prove the $\np$-completeness of the $(k,1)$-cover problem for $k=3$ and $k\geq 4$, respectively. Using the same reductions of \Cref{reductionk3} and \Cref{reductionkgeq4}, in \Cref{hardnessapprox} we prove that for every constant $k\geq 3$, the $(k,1)$-cover problem cannot be approximated within a constant factor unless $\p=\np$.

\subsubsection{Reduction for $k=3$}
\label{reductionk3}
Given an instance $(\mathcal X, \mathcal F, t)$ of SET-COVER, we construct a graph $\greductiontwo=(\vreductiontwo,\ereductiontwo)$ as follows (see \Cref{figreductionk3} for an illustration). 
\begin{enumerate}
\item Initially, $\vreductiontwo=\emptyset$ and $\ereductiontwo=\emptyset$.
    \item For every set $S_j \in \mathcal F$, we add a vertex $S_j$ to $\vreductiontwo$. We refer to such vertices as \defin{set vertices}.
    \item For every item $x_i \in \mathcal X$, we add a subgraph $I_i$ to $\greductiontwo$. Each $I_i$ is a disjoint union of $2\numitems$ isolated vertices. We refer to each  $I_i$ as an \defin{item subgraph}. We update $\vreductiontwo$ accordingly.
    \item For each $x_i \in \mathcal X$ and every $S_j \in \mathcal{F}$, if $x_i \in S_j$, we connect $S_j \in \vreductiontwo $ to all $2\numitems$ vertices of $I_i$, i.e., $\forall v \in V(I_i):\ereductiontwo \xleftarrow{} \ereductiontwo \cup \{(S_j,v)\}$.
    \item For every edge $(S_j,v)$ added in the previous step, we add a new vertex $w$ to $\vreductiontwo$ by setting $\vreductiontwo \xleftarrow[]{}\vreductiontwo \cup \{w\}$. Furthermore, we set $\ereductiontwo \xleftarrow{}\ereductiontwo \cup \{(S_j, w),(v,w)\}$. At the end of this step, all edges of $\greductiontwo$ are contained in triangles. We refer to these vertices $w$ as \defin{auxiliary vertices}.
    \item We add a vertex $P$ to $\greductiontwo$, $\vreductiontwo \xleftarrow{} \vreductiontwo \cup \{P\}$. For every item subgraph $I_i$, we connect every vertex of $I_i$ to this new vertex, i.e., $\forall I_i \forall v\in V(I_i): \ereductiontwo \xleftarrow[]{} \ereductiontwo \cup \{(v,P)\}$. We refer to vertex $P$ as \defin{the common vertex}.
\end{enumerate}

     \begin{figure}[h]
    \centering
        {\includegraphics[scale=0.6]{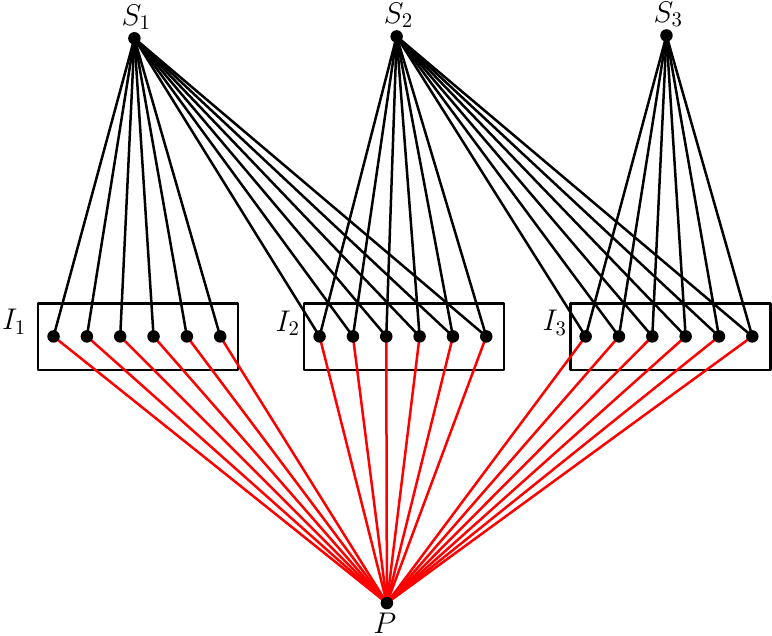}}
    \caption{An example of the reduction graph $\greductiontwo$ for $\mathcal{F}=\{S_1, S_2,S_3\}$, $\mathcal{X}=\{x_1,x_2,x_3\}$, $S_1=\{x_1,x_2\}$, $S_2=\{x_2,x_3\}$, and $S_3=\{x_3\}$. Every black edge in this graph is contained in a triangle consisting of its endpoints plus one other auxiliary vertex omitted from this figure for simplicity (see Step 5 of the construction). Therefore, only the red edges of this graph are not contained in any triangles. } 
    \label{figreductionk3}
\end{figure}
Note that all edges incident to the common vertex $P$ are 3-unsaturated. Let $x_i $ ($S_j$) be some item (set) with $x_i \in S_j$. Observe that adding the edge $(S_j, P)$ to $\greductiontwo$ saturates all edges $(v,P)$ for all $v \in V(I_i)$. Note that $\greductiontwo$ has exactly $\numsets$ set vertices, $2\numitems^2$ item vertices, one common vertex $P$, and at most $2\numitems(\numitems.\numsets)=2\numitems^2.\numsets$ auxiliary vertices. Therefore, the size of the reduction graph is polynomial in $\numitems + \numsets$, as stated in the following observation. \\

\begin{observation}
Given an instance $(\mathcal X, \mathcal F, t)$ of SET-COVER, the reduction graph $\greductiontwo$ has at most $\bigoh(\numitems^2. \numsets)$ vertices. Furthermore, due to the existence of the common vertex $P$, $\greduction$ is connected.
\label{obssizereduction}
\end{observation}
 We present the following definition.
 \begin{definition}
    Given an instance $(\mathcal X, \mathcal F, t)$ of SET-COVER, let $E'$ be a completion set of its corresponding reduction graph $\greductiontwo$. We define $R$ as the set of all non-edges of $\greductiontwo$ between the set vertices and the common vertex, i.e., {$R=\{(S_j,P)|\; S_j \in \mathcal F\}$}. If $E' \subseteq R$, we say $E'$ is \defin{a good 3-completion set} of $\greductiontwo$. 
    
 \end{definition}
For an instance $\mathcal I$ of SET-COVER, the next lemma helps us in constructing a set cover for $\mathcal I$ from a good 3-completion set of its corresponding graph $\greductiontwo$. 

\begin{lemma}
     Any good 3-completion set $E'$ of $\greductiontwo$ corresponds to a set cover of size $|E'|$ for the corresponding SET-COVER instance.
     \label{obs6}
\end{lemma}
\begin{proof}
{Let \( E' \) be a good 3-completion set of \( \greductiontwo \). For every edge \( (S_j, P) \in E' \), we select the set \( S_j \) to be included in the set cover. Since \( E' \) is a 3-completion set, the unsaturated edges connecting each item subgraph \( I_i \) to the common vertex $P$ must be covered by at least one edge \( (S_j,P)\in E'\) with $x_i \in S_j$. Therefore, the sets corresponding to the edges in \( E' \) collectively cover all items in the SET-COVER instance.}
\end{proof}
To prove the hardness result, we show that any 3-completion set $E'$ can be transformed into a good 3-completion set $E''$, where $|E''|\le |E'|$, using  \Cref{nphardconstruction2}.

\scalebox{0.91}{
\begin{minipage}{1.0\textwidth}
\begin{algorithm}[H]
    \caption{}
    \label{nphardconstruction2}
    \begin{algorithmic}[1]
            \State \textbf{Input:} $\greductiontwo$, $E'$ (a 3-completion set of $\greductiontwo$) 
            \State \textbf{Output:} $E''$, a good 3-completion set of $\greductiontwo$ with $|E''|\leq |E'|$
            \State \textbf{Initialization:} $E''\xleftarrow{}\emptyset$
            \State \textbf{Step 1:} For every item $x_i \in \mathcal X$, let $S_j \in \mathcal F$ be a set with $x_i \in S_j$. If $E'$ has an edge $(u,v)$ such that $\{u,v\}\subseteq V(I_i)$ for its corresponding item subgraph $I_i$, then add $(S_j,P)$ to $E''$, i.e., $E''\xleftarrow{}E'' \cup \{(S_j,P)\}$.
            \State \textbf{Step 2:} Add the \emph{good} subset of $E'$ to $E''$, i.e., $E'' \xleftarrow{} E'' \cup (E' \cap \{(S_j,P)|\; S_j \in \mathcal F\})$.
            \State \textbf{Step 3:} For each item $x_i \in \mathcal X$, if $I_i$ has an unsaturated edge $(u,P)\in \ereductiontwo$ (with $u \in V(I_i)$) in $\greductiontwo \cup E'' $, then add the edge $(S_j,P)$ to $E''$ for some set $S_j \in \mathcal{F}$ with $x_i \in S_j$.
            \State \Return $E''$
    \end{algorithmic}
\end{algorithm}
\end{minipage}
}
\begin{lemma}
    Let $E'$ be a 3-completion set of the described graph $\greduction$. \Cref{nphardconstruction2} returns a good 3-completion set $E''$ of $\greductiontwo$ in time polynomial in $(\numitems+ \numsets)$ with $|E''|\leq |E'|$ where $\numitems$ and $\numsets$ denote the number of items and sets of the corresponding SET-COVER instance, respectively.
    \label{lemk3con}
\end{lemma}
\begin{proof}
Since \Cref{nphardconstruction2} only adds edges of type $(S_j,P)$ and continues until all edges of $\greductiontwo$ are saturated, $E''$ is a good completion set of $\greductiontwo$. Furthermore, \Cref{nphardconstruction2} runs in time polynomial in $\numsets + \numitems$.\\
We now show that $|E''|\leq |E'|$. To show this inequality, for every edge added to $E''$, we \emph{match} it to a unique edge in $E'$ such that every edge in $E'$ is matched to at most one edge in $E''$. As a result, $|E''|\leq |E'|$.

Indeed, this matching is easy to see for Step 1 and Step 2 of \Cref{nphardconstruction2}. In Step 1, if an edge $(S_j,P)$ is added to $E''$, we match it to some edge $(u,v)\in E'$ with $\{u,v\}\subseteq V(I_i)$ and $x_i\in S_j$. In Step 2, if some edge $(S_j,P)$ is added to $E''$ (and was not added in the previous step), then we match it to its copy in $E'$, i.e., $(S_j, P) \in E'$. Therefore, in the first two steps, the edges of $E''$ are matched to unique edges from $E'$. To show this matching for Step 3, we prove the following claim.
\begin{claim}
    At the beginning of Step 3, if $\greductiontwo \cup E''$ has an unsaturated edge $(u,P)$ such that $u\in V(I_i)$ for some item $x_i$, then $E'$ has at least $2\numitems$ edges that were not matched to any edge from $E''$ in Step 1 and Step 2. 
    \label{claimunmatched}
\end{claim}
\begin{proof}
    Suppose such an edge exists. Notice that for any two distinct vertices $u,v\in V(I_i)$, we must have $(u,v)\notin E'$, because otherwise \Cref{nphardconstruction2} would have caught this edge in Step 1 and added an edge $(S_j,P)$ for $x_i \in S_j$, saturating all such edges $(u,P)$. Similarly and using Step 2, we have $(S_j,P)\notin E'$ for any $S_j \in \mathcal F$ with $x_i \in S_j$. Since $E'$ is a completion set of $\greductiontwo$, it is easy to see that $2\numitems$ edges $\{(u,P)|\;  u \in V(I_i)\}$ were covered by at least $2\numitems$ edges in $E'$, not of the types matched in the first two steps.
\end{proof}
We now conclude the proof of \Cref{lemk3con}. If, after Step 2, no such unsaturated edge $(u,P)$ exists, then we are done. If such an edge exists, then using Claim~\ref{claimunmatched}, $E'$ still has at least $2\numitems$ unmatched edges. Since we add at most $\numitems$ edges to $E''$ in Step 3 (one for each item in $\mathcal X$), we can easily match these edges to the ones described in Claim~\ref{claimunmatched}.
\end{proof}
\begin{thm}
    The $(3,1)$-cover problem is $\np$-complete.
    \label{nphardnesskeq3}
\end{thm}
\begin{proof}
    It is easy to see that the $(3,1)$-cover problem belongs to $\mathbb{NP}$ since given a set of non-edges of an input graph $G$, it can be verified in polynomial time whether that set is a 3-completion set of $G$ of size at most $t$. We now show that the $(3,1)$-cover problem is $\mathbb{NP}$-hard. Let $\mathcal{I}=(\mathcal{X}, \mathcal{F},t)$ be an instance of the SET-COVER problem. Construct the graph $\greduction$. Constructing this graph using \Cref{obssizereduction} takes polynomial time in $|\mathcal X|$ and $|\mathcal F|$. We claim that $\mathcal{I}$ has a set cover of size at most $t$ if and only if $\greductiontwo$ has a 3-completion set of size at most $t$. Indeed, if $\mathcal I$ has a set cover of size at most $t$, then we can construct a 3-completion set for $\greductiontwo$ consisting of edges $(S_j,P)$ for every set $S_j$ in this set cover. Conversely, if $\greductiontwo$ has a 3-completion set $E'$ of size at most $t$, then using \Cref{lemk3con} and \Cref{nphardconstruction2}, we can construct a good 3-completion set $E''$ with $|E''|\leq |E'|\leq t$ which corresponds to a set cover of size $|E''|\leq t$ using \Cref{obs6}. Therefore, the $(3,1)$-cover problem is $\mathbb{NP}$-hard. 
\end{proof}
\subsubsection{Reduction for $k\geq4 $}
\label{reductionkgeq4}
Given an instance $(\mathcal{X}, \mathcal{F},t)$ of the $\text{SET-COVER}$ problem and an integer $k\geq 4$, we construct a graph $\greduction= (\mathbb V, \mathbb E)$ as follows (see \Cref{fig1} for an illustration). \\
\begin{enumerate}
\item Initially, $\vreductiontwo=\emptyset$ and $\ereductiontwo=\emptyset$.
 \item For each item $x_i \in \mathcal{X}$, we add two vertices $x_i$ and $x'_i$ to $\vreductiontwo$. Moreover, we add the edge $(x_i, x'_i)$ to $\mathbb{E}$. We refer to such vertices and edges as \defin{item vertices and edges}, respectively.

    \item For each set $S_j \in \mathcal{F}$, we create a graph  $\greductionsj$ isomorphic to $K_{k-2}$ without one edge and add $\greductionsj$ to $\greduction$. Denote the missing edge of this subgraph by $(S_j,S'_j)$. We refer to such subgraphs $\greduction_j$ as \defin{set subgraphs}. We update $\vreductiontwo$ and $\ereductiontwo$ accordingly by adding the $k-2$ vertices and $\binom{k-2}{2}-1$ edges for each set subgraph to $\vreductiontwo$ and $\ereductiontwo$, respectively. 
   
    \item For each item $x_i \in \mathcal{X}$, and every set $S_j \in \mathcal{F}$ with $x_i \in S_j$, we connect $x_i$ and $x'_i$ to every vertex of $V(\greductionsj)$, i.e., we set $\ereductiontwo\xleftarrow[]{}\ereductiontwo\cup\{(u,v)|u \in \{x_i, x'_i\}, v \in V(\greductionsj)\}$. 
    \item For each edge $e_i \in \mathbb{E} \sm \{(x_1, x'_1), \dots, (x_{|\mathcal{X}|}, x'_{|\mathcal{X}|})\}$, we add a $k$-clique to $\greduction$, consisting of the endpoints of $e_i$ plus $k-2$ new vertices. We then update $\vreductiontwo$ and $\ereductiontwo$.
    
    \item We add a vertex $P$ to $\vreductiontwo$. For every set $S_i \in \mathcal F$, we add the edge $(S_i, P)$ to $\ereductiontwo$. We refer to $P$ as \defin{the common vertex} of $\greduction$. For every newly added edge $(S_i,P)$, we cover it in a $k$-clique consisting of $S_i$, $P$, and $k-2$ new vertices. This step ensures that $\greduction$ is connected. 
\end{enumerate}
It is easy to see that after the fourth step, for any $S_j \in \mathcal{F}$ and any $x_i \in \mathcal{X} \cap S_j$, the subgraph of $\greduction$ induced on $V(\greductionsj) \cup \{x_i, x'_i\}$ is a complete graph on $k$ vertices minus one edge, the one between $S_j$ and $S'_j$. Due to Step 5 and Step 6, all edges except the ones created in the second step ($\{(x_1,x'_1), \dots, (x_{|\mathcal{X}|}, x'_{|\mathcal{X}|})\}$) are covered in $k$-cliques.

The next lemma states that when $k$ is constant, the described graph can be constructed in time polynomial in $|\mathcal{F}|$ and $|\mathcal{X}|$. 
\begin{lemma}
    Let $(\mathcal{X}, \mathcal{F},t)$ be an instance of the SET-COVER problem. Then, for an integer $k \geq 4$, the procedure described above builds the graph $\greduction$ in $\mathcal{O}(|\mathcal{X}|. |\mathcal{F}|.k^{4})$ time. Furthermore, $\greduction$ is connected. 
    \label{reductiontime}
\end{lemma}
\begin{proof}
    Step 2 and Step 3 can be executed in $\mathcal{O}(\numitems)$ and $\mathcal{O}(\numsets.k^2)$ time, respectively. In Step 4, we need to check at most $\bigoh(|\mathcal{F}|.|\mathcal{X}|)$ pairs of item-sets, and at each step, we create at most $\bigoh(k^2)$ new vertices. Therefore, Step 4 takes $\mathcal{O}(|\mathcal{F}|.|\mathcal{X}|.k^2)$ time. After Step 4, there are at most $\bigoh(|\mathcal{F}|.|\mathcal{X}|.k^2)$ edges and creating a $k$-clique on each of them takes a total of $\bigoh(|\mathcal{F}|.|\mathcal{X}|.k^4)$ time in Step 5. Finally, Step~6 ensures that $\greduction$ is connected and can be done in $\mathcal{O}(\numsets .k^2)$ time.
\end{proof}
     \begin{figure}[H]
    \centering
        {\includegraphics[scale=0.51]{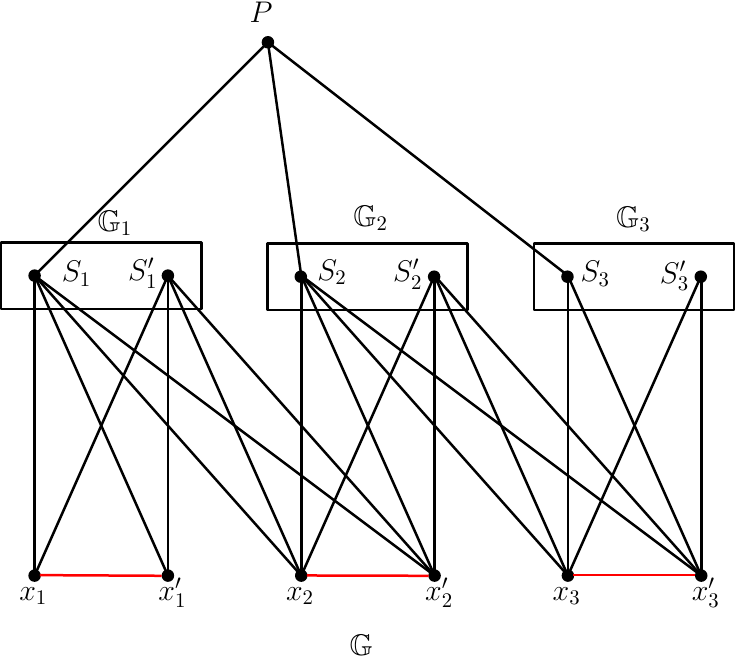}}
    \caption{An example of the reduction graph $\greduction$ for $k=4$, $\mathcal{F}=\{S_1, S_2,S_3\}$, $\mathcal{X}=\{x_1,x_2,x_3\}$, $S_1=\{x_1,x_2\}$, $S_2=\{x_2,x_3\}$, and $S_3=\{x_3\}$. Every black edge in this graph is contained in a $k$-clique consisting of its endpoints plus $k-2$ other vertices omitted from this figure for simplicity (see Step 5 and Step 6 of the construction). Therefore, only the red edges of this graph are not contained in any $k$-cliques.} 
    \label{fig1}
\end{figure}

We have the following observations regarding the graph $\greduction$.
\begin{observation}
Let $\greduction$ be the graph described above and let $R$ be the set of $k$-unsaturated edges of $\greduction$. Then, 
$R=\{(x_1, x'_1),\dots, (x_{|\mathcal{X}|}, x'_{|\mathcal X|})\}$
\label{obs1}
\end{observation}
\noindent Before describing our key lemma of this section, we present one notation.
\begin{notation}\label{vt}
Let $G=(V,E)$ be any graph. For an edge $e=(u,v) \in E$, we denote all vertices $w$ such that $w$ and $e$ form a triangle as $V_T(e)$, i.e.,
$V_T(e)\coloneqq\{w \in V| (u,w), (v,w) \in E\} $.
\end{notation}
In the next observation, we classify the triangles of each edge $(x_i,x'_i)$:
\begin{observation}
    Let $\greduction$ be the reduction graph. For an item $x_i\in \mathcal X$, let $\{S_1, \dots S_m \}\subseteq \mathcal F$ be the sets in $\mathcal F$ that cover $x_i$, i.e., $x_i\in S_1\cap\dots\cap S_m$. We have $V_T((x_i, x'_i)) \subseteq \bigcup_{j=1}^{m}V(\greduction_j)$ in $\greduction$, where $\mathbb{G}_1,\dots,\mathbb{G}_m $ are the corresponding set subgraphs (see Step 3 of the construction) of $S_1,\dots ,S_m $, respectively. 
    \label{obs2}
\end{observation}
In other words, \cref{obs2} states that in $\greduction$, the third vertex of any triangle containing $(x_i, x'_i)$ must necessarily come from some set subgraph $\greductionsj$ (with $x_i \in S_j$) as created in Step 3.
We present a definition.
\begin{definition}
    Let $E'$ be a $k$-completion set for the graph $\greduction$. Then, $E'$ is a \defin{good $k$-completion set} if {$E'\subseteq \{(S_1,S'_1), \dots, (S_{|\mathcal{F}|}, S'_{|\mathcal{F}|})\}$}. 
    \label{goodcompletionsetdefinition}
\end{definition}
  Similar to \Cref{obs6}, we have the following: 
  \begin{lemma}
     Any good completion set $E'$ of $\greduction$ corresponds to a set cover of size $|E'|$ for the corresponding SET-COVER instance.
     \label{obs6prime}
\end{lemma}
\begin{lemma}
Let $\greduction$ be the described graph. Then, any $k$-completion set $E'$ of $\greduction$ can be transformed into a good $k$-completion set $E''$ with $|E''|\leq |E'|$. Furthermore, this transformation can be done in $\bigoh(|E'| +|\mathcal{F}|^2. k^2 +|\mathcal{X}|^2.|\mathcal{F}|.k^4)$ time for any $k\geq 4$.
    \label{lemk4geqcon}
\end{lemma}
\begin{proof}
    This construction is done in steps and described in greater detail in \cref{nphardconstruction}. We first show the time complexity. Trivially, Line 4 can be done $\bigoh(|E'|)$ time. In Lines 5 to 9, we check all pairs of set subgraphs and there are $\bigoh(|\mathcal{F}|^2)$ pairs in total, where each pair can be checked in $\bigoh(k^2)$ time. In Lines 10 to 20, we check $\bigoh(|\mathcal{X}|)$ item edges and handle each edge in $\bigoh(|\mathcal{X}|.|\mathcal{F}|.k^4)$ time (in the worst case, we have to check every vertex of $\greduction$). Moreover, since \Cref{nphardconstruction} only adds edges of type $(S_j,S'_j)$ and continues until all edges of $\greduction$ are saturated, it is easy to see that $E''$ is a good $k$-completion set of $\greduction$.

    To show that $|E''|\leq |E'|$, similar to the proof of \cref{lemk3con}, for every edge added to $E''$, we match it to a unique edge in $E'$ such that every edge in $E'$ is matched to at most one edge in $E''$. For Line 4, this matching is easy to see, for every such edge $(S_j, S'_j)$ added to $E''$, we match it to its copy in $E'$ ($(S_j,S'_j)\in E'$). Similarly, if an edge $(S_j, S'_j)$ is added to $E''$ (and was not added in previous steps) in Line 7, then we can match $(S_j,S'_j)$ to $(u,v)\in E'$ (as described in Line 6). Because of the way we pick $(u,v)\in E'$ in Line 6, it could not have been matched to another edge $(S_{i}, S'_{i})\in E''$ before $(S_j, S'_j)\in E''$ (with $i\neq j$). 

    Before describing the matching for Lines 10 to 20, we show the following claim. 
        \begin{claim}
        As long as we have a $k$-unsaturated item edge $(x_i,x'_i)$ in $\greduction \cup E''$ in Line 10, there exists a vertex $v \in   V_T((x_i, x'_i))$ in $\greduction \cup E'$ with $|\mathbb{E} \cap \{(x_i, v), (x'_i,v)\}|\leq 1$ in Line 11. Moreover, the edges of $\{(x_i,v), (x'_i,v)\} \cap E'$ have not been matched to any edges in $E''$ in the previous steps.
        \label{claim1}
    \end{claim}
    \begin{proof}
 Consider the initial $k$-completion set $E'$ and the first such unsaturated edge $(x_i, x'_i)$ (we will address subsequent edges later). If the edges in $E'$ create any new triangles-i.e., triangles that exist in $\greduction \cup E'$ but not in $\greduction$-containing $(x_i, x'_i)$, then such a vertex $v$ trivially exists in $\greduction \cup E'$. 

Now, assume to the contrary that such a vertex $v$ does not exist, i.e., $V_T((x_i, x'_i))$ in $\greduction$ is the same as $V_T((x_i, x'_i))$ in $\greduction \cup E'$. This would imply that $E'$ has created a $k$-clique containing $x_i, x'_i$, and $k-2$ other vertices from $V_T((x_i, x'_i))$ in $\greduction$. Observe that these $k-2$ vertices from $V_T((x_i, x'_i))$ must form a $(k-2)$-clique in $G\cup E'$ and let $C$ denote this $(k-2)$-clique. By \cref{obs2}, there are two possible cases for $C$:\\ 
        \noindent \textbf{Case I:} Assume that $V(C)=V(\greduction_j)$ for some set subgraph $\greduction_j$ with $x_i \in S_j$. This case implies that $(S_j,S'_j)\in E'$. However, this leads to a contradiction since \cref{nphardconstruction} would have caught this edge in Line 4 and saturated $(x_i,x'_i)$ by adding $(S_j, S'_j)$ to $E''$ (since $x_i \in S_j$). \\
           \noindent \textbf{Case II:}
           Assume $V(C)\subseteq V(\greduction_j) \cup V(\greduction_{\ell})$ for two set subgraphs $\greduction_j$ and $\greduction_{\ell}$ with $j < \ell$ and $x_i \in S_j \cap S_{\ell}$. This case implies that $E'$ contains at least one edge between the vertices of $\greduction_j$ and $\greduction_{\ell}$. This is a contradiction, because \cref{nphardconstruction} would have caught this edge in Line 5 and saturated $(x_i,x'_j)$ by adding $(S_j,S'_j)$ to $E''$ in Line 7 (since $x_i \in S_j$).
           
           Therefore, such a vertex $v$ exists for the first $k$-unsaturated item edge $(x_i,x'_i)$. Moreover, in Lines 4 to 9, no edge from $E'$ with an endpoint in $\{x_i, x'_i\}$ is ever matched to an edge in $E''$. Thus, prior to processing $(x_i, x'_i)$ in Line 10, the edges in $\{(x_i,v), (x'_i,v)\} \cap E'$ are not matched to any edge in $E''$. Therefore, for the first $k$-unsaturated item edge $(x_i,x'_i)$ of Line 10,  we can match the edges added to $E''$ in Lines 15 and 18 to the ones from $E'$ described in Lines 14 and 17, respectively.
           
For any subsequent $k$-unsaturated edge $(x_i, x'_i)$ processed in Line 10, we can apply the same reasoning to conclude that $V_T((x_i, x'_i))$ in $\greduction$ is not equal to $V_T((x_i, x'_i))$ in $\greduction \cup E'$. Moreover, the edges in $\{(x_i, v), (x'_i, v)\} \cap E'$ have not been matched to any edges in $E''$ during the previous steps. This remains true even after Line 15 of a previous iteration. If $(x_i, x'_i) = (x_{\ell}, x'_{\ell})$ in Line 13 for some previous iteration, then $(x_i, x'_i)$ would not be $k$-unsaturated in the current iteration, as it would have already been saturated in Line 15 of that previous iteration.

This completes the proof of \cref{claim1}.  
    \end{proof}
Using \cref{claim1}, we can wrap up the proof of \cref{lemk4geqcon}. Note that we can match the edges added to $E''$ in Lines 15 and 18 to the ones from $E'$ described in Lines 14 and 17, respectively.
\end{proof}

   \begin{algorithm}[H]
    \caption{The construction method of \cref{lemk4geqcon}}
    \label{nphardconstruction}
    \begin{algorithmic}[1]
            \State \textbf{Input:} $\greduction$, $E'$ (a $k$-completion set of $\greduction$) 
            \State \textbf{Output:} $E''$, a good $k$-completion set of $\greduction$ (see \cref{goodcompletionsetdefinition}) with 
            $|E''|\leq |E'|$
            \State \textbf{Initialization:} $E''\xleftarrow{}\emptyset$
            \State Add the \emph{good} subset of $E'$ to $E''$, i.e., $E'' \xleftarrow[]{} E' \cap \{(S_1, S'_1), \dots, (S_{|\mathcal F|}, S'_{|\mathcal{F}|})\}$
            \For{ any two sets $S_j, S_{\ell} \in \mathcal{F}$ with $j < \ell$}
            \If{exists $e=(u,v)$ in $E'$ with $u \in V(\greduction_j)$ and $v\in V(\greduction_{\ell})$}
            \State  Add $(S_j, S'_j)$ to $E''$, i.e., $E'' \xleftarrow{} E'' \cup \{(S_j,S'_{j})\}$.
            \EndIf
            \EndFor
\While{there exists a $k$-unsaturated item edge $(x_i,x'_i)$ in $\greduction \cup E''$}
\State Let $v \in   V_T((x_i, x'_i))$ in $\greduction \cup E'$ with $|\mathbb{E} \cap \{(x_i, v), (x'_i,v)\}|\leq 1$
\State Let $S_i\in \mathcal F$ be a set with $x_i \in S_i$
\If{$v\in \{x_{\ell}, x'_{\ell}\}$ for some other $k$-unsaturated item edge $(x_{\ell}, x'_{\ell} )$ in $\greduction \cup E''$}
\State In this case, we have $\{(x_i,v), (x'_i,v)\} \subseteq E'$, let $S_{\ell}\in \mathcal F$ be a set with $x_{\ell} \in S_{\ell}$
\State Add $(S_i,S'_i)$ and $(S_{\ell}, S'_{\ell})$ to $E''$, i.e.,  $E''\xleftarrow[]{} E'' \cup \{(S_i,S'_i), (S_{\ell}, S'_{\ell})\}$. This saturates $(x_i, x'_i)$ and $(x_{\ell}, x'_{\ell})$ in $\greduction \cup E''$.
  \Else
        \State In this case, there exists an edge $e \in E' \cap \{(x_i,v), ( x'_{i},v)\}$  
      \State Add $(S_i,S'_i)$ to $E''$, i.e., $E'' \xleftarrow{} E'' \cup \{(S_i,S'_i)\}$
\EndIf

\EndWhile
    \end{algorithmic}
\end{algorithm}
\begin{thm}
    The \mbox{$(k,1)$-cover} problem is $\np$-complete for every constant $k \geq 4$.
    \label{nphardnesskgeq}
\end{thm}
\begin{proof}
 It is easy to see that for any constant $k\geq 4$, the $(k,1)$-cover problem belongs to $\mathbb{NP}$. Furthermore, in time polynomial in $\numitems+ \numsets$, we can build the graph $\greduction$ using \cref{reductiontime} and the fact that $k$ is constant.

The $\np$-hardness proof is similar to the proof of \Cref{nphardnesskeq3}, and we skip the details for brevity. Note that every set cover of the corresponding SET-COVER instance corresponds to a good completion set for $\greduction$ of the same size. Moreover, using \Cref{nphardconstruction}, every completion set $E'$ (with $|E'|\leq t$) of $\greduction$ can be converted into a good completion set $E''$ with $|E''|\leq |E'|\leq t$, which corresponds to a set cover of size at most $t$ for the underlying SET-COVER instance. 
\end{proof}
\subsubsection{Inapproximability of the $(k,1)$-Cover Problem for General Graphs}
\label{hardnessapprox}
In this section, we prove that it is hard to approximate the $(k,1)$-cover problem within a constant factor for any constant $k\geq 3$.

We first restate a result on the inapproximability of SET-COVER.

\begin{lemma}\cite[Corollary 4]{dinur2014analytical}

For every $\varepsilon >0$, it is $\np$-hard to approximate SET-COVER by a $((1-\varepsilon). \ln |\mathcal{X}|)$ factor.
    \label{setcoverhard}
\end{lemma}
\begin{thm}
     For any constant $k\geq 3$, it is $\np$-hard to approximate the $(k,1)$-cover problem within a factor of $c$ for any constant $c>1$. 
    \label{thm1}
\end{thm}
\begin{proof}

    For the sake of contradiction, suppose $\operatorname{\mathbb{P}} \neq \operatorname{\mathbb{NP}}$ and there exists a polynomial-time $c$-approximation algorithm $A$ for the $(k,1)$-cover problem for some constants $c>1$ and $k\geq3$, where $c\in \mathbb{R}
$ and $k \in \mathbb{Z}$. For any instance $\mathcal{I}=(\mathcal{X}, \mathcal{F},t)$ of the SET-COVER problem, let $\opt_S$ be the size of the optimal set cover, and for its corresponding reduction graph $\greductiontwo$, let $\opt_G$ denote the size of its optimal completion set. Note that $\opt_S=\opt_G$. Every set cover for $\mathcal{I}$ corresponds to a completion set for $\greductiontwo$, so $\opt_G \leq \opt_S$. On the other hand, we also have $\opt_S\leq \opt_G$, because otherwise we would have $\opt_G <\opt_S$ and using \Cref{nphardconstruction2}, \Cref{nphardconstruction}, \Cref{obs6}, \Cref{lemk3con}, \Cref{obs6prime}, and \Cref{lemk4geqcon}, $\mathcal I$ would have a set cover of size strictly less than $\opt_S$, a contradiction. Therefore, we have $\opt_G=\opt_S$.

    Given algorithm $A$, we now show how to approximate any instance of SET-COVER within a factor of $c$ in polynomial time. Given any such instance $(\mathcal X, \mathcal{F})$, we construct the reduction graph $\greduction$ in time polynomial in $\numsets +\numitems $ (\Cref{obssizereduction} and \cref{reductiontime}). We run $A$ on $\greduction$ that gives us a $k$-completion set $E'$ with $|E'|\leq c\times \opt_G$ in time polynomial in the size of $\greduction$ and hence in $|\mathcal{X}|$ and $|\mathcal{F}|$. Using \Cref{nphardconstruction2} (for $k=3$) and \Cref{nphardconstruction} (for $k\geq 4$), we can transform $E'$ to another good $k$-completion set $E''$ with $|E''|\leq |E'|$ in time polynomial in $|\mathcal{X}|$ and $|\mathcal{F}|$ (note that $|E'|$ is polynomial in $\numitems+ \numsets$). However, $E''$ is a good $k$-completion set that corresponds to a set cover of size $|E''|$ with
    $$|E''|\leq|E'|\leq c\times \opt_G =c\times \opt_S$$
    Therefore, this results in a polynomial-time constant-factor approximation algorithm for SET-COVER. However, from \Cref{setcoverhard}, we know that for every $\varepsilon >0$, it is $\np$-hard to devise an $((1-\varepsilon). \ln |\mathcal{X}|)$-approximation to SET-COVER, a contradiction.
\end{proof}
\subsection{$\np$-Hardness for Trees}
\label{nphardnesstrees}
In this section, we prove that for general values of $k$, the $(k,1)$-cover problem is $\np$-hard, even when its input is a spider. A \defin{spider} is a tree with exactly one vertex of degree at least 3, to which we refer as the \defin{center} of the spider. A spider has \defin{legs} which are paths of varying sizes meeting at the center of the spider. We reduce the $\np$-complete problem of 3-PARTITION to our problem on spiders. We use some ideas from~\cite{reductionideas} and extend them with some structural properties of extremal graphs with $(k,1)$-covers.

We start by formally defining the 3-PARTITION problem. 
\begin{definition}
   Given a multiset $S=\{a_1,\dots, a_{3p}\} $ of $3p$ positive integers, an integer $s>0$ such that each $a_i$ satisfies $s/4< a_i<s/2$ and $\sum_{i=1}^{3p}a_i=sp$, an instance $\mathcal I=(S,s) $ of 3-PARTITION asks whether there exists a partition of $S$ into $p$ subsets of size exactly three, such that each subset sums up to $s$.
\end{definition}
3-PARTITION is strongly $\np$-complete, i.e., it remains $\np$-complete even if all integers are polynomially bounded in the size of the input~\cite{threepartitionhardness}. For completeness, we also define the following decision problem for trees. 
\begin{definition}
    Given a tree $T=(V,E)$, two integers $k\leq |V|$ and $t\geq 0$, an instance $\mathcal I= (T, k, t)$ of TREE-COMPLETION asks whether $T$ has a $k$-completion set of size $t$.
\end{definition}
We now describe a reduction from 3-PARTITION to TREE-COMPLETION. For an instance $\mathcal{I}=(S,s)$ of 3-PARTITION, we construct a spider $\treduction=(\vreductiontwo, \ereductiontwo)$ as follows.
\begin{enumerate}
    \item Initially, we set $\vreductiontwo=\emptyset$ and $\ereductiontwo=\emptyset$.
    \item We add a vertex $r$ to $\vreductiontwo$. $r$ will be the center of this spider.
    \item For each $a_i \in S$, we add a leg with $a_i$ edges to $r$ and update $\vreductiontwo$ and $\ereductiontwo$ accordingly. 
\end{enumerate}
The following observation states that the described spider $\treduction$ can be constructed in time polynomial in the size of the corresponding 3-PARTITION instance. 
\begin{observation}
    For an instance $\mathcal I=(S,s)$ of 3-PARTITION, the described tree $\treduction=(\vreductiontwo,\ereductiontwo)$ is a spider with $|\vreductiontwo|=sp+1$ and $|\ereductiontwo|=sp$. Furthermore, since $s$ and $p$ are polynomially bounded in the size of the input, the size of $\treduction$ is polynomial in the size of $\mathcal I$.
    \label{obstreduction}
\end{observation}
To prove our main hardness result of this section, we need the following structural result. This result can be deduced by analyzing Theorem 1 in \cite{chakraborti2024sparse} for the case of $r=k-1$.  
\begin{lemma}
       Let $k\ge 3$ and $G=(V,E)$ be a connected $n$-vertex graph with a \mbox{$(k,1)$-cover} such that $n-k = q (k-1) +r$ where $q\ge 0$ and $r= k-1$. If $|E|=(q+2) \binom{k}{2}$, then there exist $q+2$ subgraphs $C_1, \dots, C_{q+2}$ of $G$ such that 
       \begin{enumerate}[(i)]
           \item Each $C_i$ is isomorphic to a $k$-clique.
           \item These subgraphs are pairwise edge-disjoint, i.e., for any two distinct $C_i$ and $C_j$ we have $E(C_i) \cap E(C_j)=\emptyset $.
           \item Every edge of $G$ belongs to some $C_i$, i.e., $\bigcup_{i=1}^{q+2}E(C_i)=E$.
       \end{enumerate}
       \label{lemstructureopt}
       \end{lemma}
The following lemma is essential for proving the $\np$-hardness. 
\begin{lemma}
    For an instance $\mathcal I=(S,s)$ of 3-PARTITION, let $\treduction=(\vreductiontwo,\ereductiontwo)$ be the described spider. Then, the following statements are equivalent.
    \begin{enumerate}[(i)]
        \item $\mathcal I$ is a YES instance of 3-PARTITION. 
        \item $\treduction$ has an $(s+1)$-completion set $E'$ of size $\frac{ps(s-1)}{2}$.
        \item $\ereductiontwo$ can be partitioned into $p$ edge-disjoint trees, such that each tree has exactly $s$ edges.
    \end{enumerate}
    \label{lemtreduction}
\end{lemma}
\begin{proof}
    We show $\text{(i)}\implies\text{(ii)}\implies\text{(iii)}\implies\text{(i)}$.\\\\
    \noindent \textbf{$\text{(i)} \implies\text{(ii)}$:} Suppose $S$ can be partitioned into $p$ subsets of size three such that each subset sums up to $s$. Then, for each such subset, we pick three legs of $\treduction$ corresponding to the three integers of that subset. These three legs with the center $r$ make up a spider with exactly $s$ edges. We can convert this smaller spider into an $(s+1)$-clique by adding exactly $\binom{s+1}{2}-s=\frac{s(s-1)}{2}$ edges to $E'$ (initially, $E' \xleftarrow[]{}\emptyset$). Since there are $p$ subsets in total, there are $p$ spiders with three legs that collectively cover $\ereductiontwo$. Converting each such spider into an $(s+1)$-clique can be done by adding exactly $\frac{ps(s-1)}{2}$ edges to $E'$ in total, making $E'$ an $(s+1)$-completion set of size $\frac{ps(s-1)}{2}$ for $\treduction$.\\\\
     \noindent \textbf{$\text{(ii)} \implies\text{(iii)}$:} Suppose such an $(s+1)$-completion set $E'$ exists, and let $\treduction'\coloneqq \treduction \cup E'$. $\treduction'$ has an $(s+1,1)$-cover, with $|V(\treduction')|=sp+1$ and $E(\treduction')=sp+ \frac{ps(s-1)}{2}=\frac{ps(s+1)}{2}$. Furthermore, we have 
     \begin{equation}
     sp+1- (s+1)=(p-2)s+s
         \label{khoshgel}
     \end{equation}
     We now apply \Cref{lemstructureopt}. In \Cref{lemstructureopt}, set $k=s+1$, $n=sp+1$, $q=p-2$ and $r=s$. Since $\treduction'$ has an $(s+1,1)$-cover with exactly $(q+2) \binom{k}{2}=p\binom{s+1}{2}=\frac{ps(s+1)}{2}$ edges, it follows that $\treduction'$ has $p$ $(s+1)$-cliques $C_1, \dots , C_p$ with the properties described in \Cref{lemstructureopt}(i)-(iii). Using \Cref{lemstructureopt}(ii) and \Cref{lemstructureopt}(iii), these cliques define an edge partition of $\treduction'$. For any such $C_i$, define $x_i =|E(C_i)\cap E(\treduction)|$. Since $|V(C_i)|=s+1$ for all $C_i$, we have $0\leq x_i\leq s$ and
     \begin{equation}
     |E'|=\frac{ps(s-1)}{2}=\sum_{i=1}^{p}\binom{s+1}{2}- x_i \text{, subject to }0\leq x_i\leq s \text{ for all }i.
         \label{khoshgel2}
     \end{equation}
     However, \eqref{khoshgel2} holds if and only if $x_i=s$ for all $i$, i.e., when each $C_i$ holds exactly $s$ edges of $\treduction$. Since $x_i=s$ and $|V(C_i)|=s+1$, it follows that for each $C_i$, $C_i\cap \treduction$ is a sub-tree of $\treduction$ with $s$ edges. Therefore, $\forall i\in \{1, \dots, p\} C_i \cap \treduction$ defines a partition of $\ereductiontwo$ into $p$ edge-disjoint sub-trees with $s$ edges each. \\\\
     \noindent \textbf{$\text{(iii)} \implies\text{(i)}$:} Suppose such an edge partition exists. First, observe that since each leg $L_i$ of $T$ has strictly less than $s$ edges, all edges of every leg $L_i$ of $\treduction$ must belong to the same tree of the partition. Moreover, every tree of this partition must consist of exactly three legs of $\treduction$, because each $a_i$ satisfies $s/4<a_i<s/2$. It follows that the edges of $\treduction$ can be partitioned into $p$ trees such that each tree has exactly $s$ edges and three legs of $\treduction$. This edge partition implies that $\mathcal{I}$ is a YES-instance of 3-PARTITION. 
     \end{proof}
 From \Cref{obstreduction} and the equivalence of \Cref{lemtreduction}(i) and  \Cref{lemtreduction}(ii), the $\np$-hardness of TREE-COMPLETION follows. 
\begin{thm}
    TREE-COMPLETION is $\np$-hard, even when the input tree is a spider. 
    \label{thmhardnesstree}
\end{thm}
\section{An Optimal Algorithm for Chordal Graphs for the $(3,1)$-Cover Problem}
\label{secchordal}
This section presents an optimal algorithm for the $(3,1)$-cover problem on chordal graphs. For convenience, throughout this section, we refer to 3-unsaturated edges and 3-completion sets as unsaturated edges and completion sets, respectively.  

\subsection{An Optimal Algorithm for the $(3,1)$-Cover of Trees}
\label{sectrees}
Before presenting our main result, we briefly describe how we can optimally solve the $(3,1)$-cover problem when the input graph is a tree. We will later use the algorithm for trees to solve the problem for chordal graphs. 

We begin by restating a known result.
\begin{lemma}    \label{theoremghabli}
\cite[Theorem 8]{BurkhardtFH20} Let $M_{n,\ell}$ be the minimum number of edges in a connected graph on $n$ vertices with a $(3,\ell)$-cover, we have
\[
        (n-1)\left(1+\frac{\ell}{ 2}\right) \leq M_{n, \ell} \leq n\left(1+\frac{\ell} {2}\right)+\Theta\left(\ell^2\right).
\]
\end{lemma}
\Cref{theoremghabli} implies the following corollary, giving us a lower bound on the size of any optimal $(3,k-2)$-completion set of an $n$-vertex tree for $k\geq 3$. 
\begin{corollary}
    Let $T$ be any tree with $|V(T)|=n$. Let $\opt$ be the size of an optimal $(3,k-2)$-completion set of $T$ for $k\geq 3$. Then, we have 
    \begin{equation}
    (n-1) \bigg(\frac{k-2}{2}\bigg)\leq \opt
        \label{optlowerbound}
    \end{equation}
    \label{corol}
\end{corollary} 
\Cref{corol} results from the simple fact that any $n$-vertex tree has $n-1$ edges.

 Let us denote the $n$-vertex path by $P_n$. We show that the edges of every $n$-vertex tree $T$ ($n\geq 3$) can be partitioned into $\lceil\frac{n-1}{2}\rceil$ sub-trees, such that all but at most one sub-tree is isomorphic to $P_3$ as the following lemma states.
\begin{lemma}\label{lem:tree(3,1)}
Let $T = (V, E)$  be a tree with $|V| = n \geq 3$. There exists an algorithm running in $\mathcal{O}(n)$ time that partitions the edges of $E$  into  $\lceil{\frac{n-1}{2}}\rceil$ sub-trees of  $T$ with the following structure. If $|E| = n-1$ is even, all these sub-trees are isomorphic to $P_3$. If $|E|$ is odd, all but one sub-trees are isomorphic to $P_3$, and the remaining sub-tree is isomorphic to $K_2$. 
\label{lemp3decompositionalg}
\end{lemma}
\begin{proof}
    Suppose $T$ is rooted at an arbitrary node $r \in V$. Create an array $A$ of length $d+1$, where $d$ is the depth of $T$ with respect to $r$. For any $0\leq i\leq d$, $A[i]$ holds a linked list containing all vertices at depth $i$ in $T$.

    The algorithm proceeds in iterations. For every iteration $j\geq 1$, we find the deepest leaf $v$ of $T$. This leaf can be found in the last non-empty cell of $A$. We then extract a sub-tree $T_j$ of $T$ in the following way. Let $u$ be the parent of $v$ in $T$. If $v$ has a sibling $v_1$, we set $T_j \xleftarrow{} T[\{u,v,v_1\}]$. If $v$ has no siblings and $T$ has at least three nodes, we set $T \xleftarrow{} T[\{u,v,w\}]$ where $w$ is the parent of $u$. If $T$ has exactly two nodes left, we set $T_j\xleftarrow{}T[\{v,u\}]$. 

    It is easy to see that $T \sm E(T_j)$ has at most one non-trivial component. 
    There may be some nodes $u$ with $d_T(u)>0$ and $d_{T \sm E(T_j)}(u)=0$. These nodes are marked as \textit{deleted}.
     We then set $T$ to be the only non-trivial component of $T\sm E(T_j)$. If $T \sm E(T_j)$ does not have any non-trivial components, we terminate the algorithm. 
    
In subsequent iterations, whenever a node \( u \) is extracted from the last non-empty cell of $A$, we check if \( u \) is marked as deleted. If so, \( u \) is discarded, and the process continues until a non-deleted node is extracted. It can be verified that as long as $T$ has at least three vertices remaining, $T_j$ is isomorphic to $P_3$. Therefore, if $n-1$ is even, we have $T_j \cong P_3$ for $\frac{n-1}{2}$ many $j$. If $n-1$ is odd, we have $T_j \cong P_3$ for $\lfloor \frac{n-1}{2}\rfloor $ many $j$, and $T_j\cong K_2$ for the last iteration. Regarding the time complexity, for finding sub-trees $T_j$ we visit every edge of $T$ exactly once; therefore, we spend $\mathcal{O}(n)$ time in total for finding all sub-trees. At every iteration $j$, we need to check a constant number of nodes to determine whether they need to be marked as deleted. Moreover, every deleted node is discarded at most once, and updating $A$ can be done in $\bigoh(n)$ time in total. Since we never increase any node's depth, finding the last non-empty cell of $A$ can be done in $\mathcal{O}(n)$ time in total since $A$ has length at most $n$.
\end{proof}

\;\\
We now show that a tree's optimal $(3,1)$-cover can be computed efficiently. 
\begin{proposition}
Let $T = (V, E)$  be a tree with $|V| = n \geq 3$. In $\mathcal{O}(n)$ 
 time, we can solve the $(3,1)$-cover problem optimally for $T$ by producing a completion set of size $\lceil\frac{n-1}{2} \rceil$.
\label{proptreeoptimal}
\end{proposition}
\begin{proof}
Use the algorithm of \Cref{lemp3decompositionalg} to find a partition of $E$ into many copies of $P_3$ and at most one $K_2$ (if \( |E| \) is odd). For each sub-tree isomorphic to \( P_3 \), we can add the missing edge to the completion set, transforming the \( P_3 \) into a \( K_3 \). For the sub-tree isomorphic to \( K_2 \), we can add one edge to the completion set that will cover the edge in this sub-tree within a triangle. It is easy to see that the described algorithm adds exactly \( \lceil \frac{n-1}{2} \rceil \) edges, and it is optimal, using the bound in \eqref{optlowerbound}.\end{proof} 
\subsection{An Optimal Algorithm for the $(3,1)$-Cover of Chordal Graphs} 
\label{secchordalmain}
We now describe our algorithm for chordal graphs. Let $G=(V,E)$ be a connected chordal graph on at least three vertices. Note that if an edge $e$ is not a bridge of a chordal graph $G$, it must belong to a cycle and, consequently, a triangle. We have the following notation.
\begin{notation}
    Let $G=(V,E)$ be a chordal graph. Denote by $T_1, \dots, T_c$ the maximal connected subgraphs on the bridges of $G$. 
    \label{treesofchordaldefinition}
\end{notation}
For convenience, we refer to these subgraphs as \defin{the trees of $G$}, see \Cref{chordalfig}. 

  \begin{figure}[tbh]
    \centering
        {\includegraphics[scale=0.80]{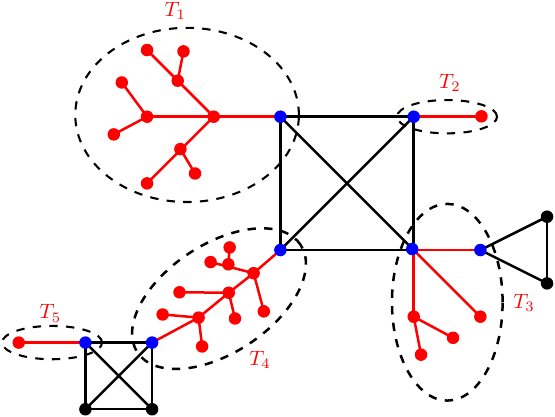}}
    \caption{An example of a chordal graph and its trees (depicted in red). The outer and boundary vertices are depicted in black and blue, respectively.} 
    \label{chordalfig}
\end{figure}
For a chordal graph $G$, we call the 

vertices that belong to no trees \defin{the outer vertices}. Moreover, we refer to the vertices of $G$ that belong to a tree and are incident to at least one non-bridge as the \defin{boundary vertices}. Our algorithm for chordal graphs is described in \Cref{chordalalgorithm}.
\begin{algorithm}[H]
    \caption{An optimal algorithm for chordal graphs}
    \label{chordalalgorithm}
    \begin{algorithmic}[1]
\State \textbf{Input:} A connected chordal graph $G=(V,E)$ with $|V| \geq 3$
\State \textbf{Output:} An optimal completion set $E'$ of $G$
\State \textbf{Initialization:} Find the trees $T_1, \dots, T_c $ of $G$ (see \Cref{treesofchordaldefinition}), set $E' \coloneqq \emptyset$.
\State For every tree $T_i$ with $|V(T_i)|\geq 3$, convert it to graph with a $(3,1)$-cover using the algorithm of \Cref{proptreeoptimal}. Update $E'$ accordingly. 
\State For any tree $T_i$ isomorphic to $K_2$, let $V(T_i)=\{u,v\}$ and without loss of generality, let $u$ be a boundary vertex of $G$. Add $(v,w)$ to $E'$, where $w\neq v$ is a neighbour of $u$ in $G$.
\State \Return $E'$
    \end{algorithmic}
\end{algorithm}
\;\\
We present a notation before proving the optimality of \Cref{chordalalgorithm}. 
\begin{notation}
Let $E'$ be any completion set of $G$. We denote by $\edgesbetweentitj$ the set of edges of $E'$ between the vertices of $T_i$ and $T_j$, i.e., $\edgesbetweentitj= \{(u,v) \in E'| u\in V(T_i), v \in V(T_j)\text{, where }i\neq j\}$. The edges in $\edgesbetweentitj$ are referred to as \defin{cross edges}.
\end{notation}

\begin{lemma}
Let \( G = (V,E) \) be any connected graph on at least three vertices, and let \( E' \) be any optimal completion set of \( G \). For every edge \( e' \in E' \), there exists a triangle in \( G \cup E' \) containing $e'$ such that at least one edge of this triangle is unsaturated in $G$ (from the set \( E \)).
\label{lem triangle necessary}
\begin{proof}
For the sake of contradiction, assume that there exists a non-empty subset \( S \subseteq E' \) of edges that do not participate in any such triangles. Let \( E'' = E' \sm S \). It can be seen that \( G \cup E'' \) has a \((3,1)\)-cover with \( |E''| < |E'| \), which contradicts the optimality of \( E' \).
\end{proof}
\end{lemma}

To prove the correctness of \Cref{chordalalgorithm}, we show that for a chordal graph $G$, any optimal completion set $E'$ can be transformed into another completion set $E''$ of the type described in \Cref{chordalalgorithm}, i.e., the edges of each tree $T_j$ with $|V(T_j)|\geq 3$ are only covered by edges completely within $T_j$. Moreover, the edge in any $T_j$ with $|V(T_j)|=2$ is covered by exactly one edge, as described in Line 5 of \Cref{chordalalgorithm}.

We present this transformation in the next two lemmas. Before describing these lemmas, we first provide an example to motivate them. The example is shown in \Cref{examplechordal} for a chordal graph $G$ with two trees $T_1$ and $T_2$. Two completion sets $E'$ and $E''$ are depicted in \Cref{examplechordal}(a) and \Cref{examplechordal}(b), respectively. Both completion sets are optimal with $|E'|=|E''|=4$; however, $|E'_{1,2}|=2$ and $|E''_{1,2}|=0$. \Cref{lem key alg} provides a procedure to transform $E'$ into $E''$ by moving the edges $(w_1, b_2)$ and $(w_2, b_1)$ within $T_1$ and $T_2$, respectively. 
Before describing the modification process, we make an observation on an edge between the boundary vertices of $T_i$ and $T_j$ in a chordal graph.
\begin{observation}
If there is an edge $e=(u,v)\in E$ between $T_i$ and $T_j$ such that $u\in V(T_i)$ and $v\in V(T_j)$, then $u$ and $v$ are boundary vertices of $T_i$ and $T_j$, respectively. 
\label{niceobsboundary}
\end{observation}
For instance, in the graph of \Cref{examplechordal}, $(b_1,b_2)$ serves as the edge described in \Cref{niceobsboundary} for two boundary vertices $b_1 \in V(T_1)$ and $b_2 \in V(T_2)$.

 We consider multiple cases for the aforementioned transformation of edges in $E'$.\\\\
    \noindent\textbf{Case I:} Edges that lie between two distinct trees $T_i$ and $T_j$ with $|V(T_i)|\geq 3$ and $|V(T_j)|\geq 3$\\\\
\noindent\textbf{Case II:} Edges that lie between two distinct trees $T_i$ and $T_j$ with $|V(T_i)|=2$\\\\
\noindent \textbf{Case III:} Edges that lie between a tree $T_i$ and some outer vertex $u$.\\ \\

We start by considering Case~I in \cref{lem key alg} and explain how other cases can be handled similarly. 
\begin{figure}[H]
    \centering
    \subfloat[$G \cup E'$]
    {{\includegraphics[scale=0.80]{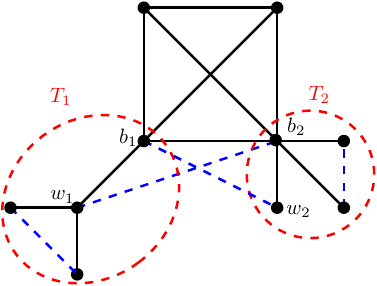}}}
    \qquad
    \subfloat[$G\cup E''$]
    {{\includegraphics[scale=0.80]{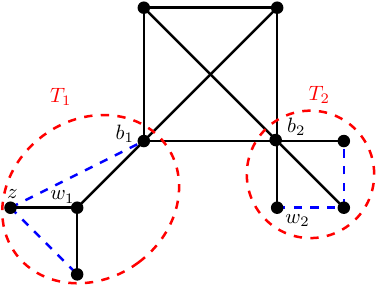}}}
    \caption{An example of a chordal graph $G$ with trees $T_1$ and $T_2$, and two completion sets $E'$ (a) and $E''$ (b). Solid edges belong to $G$, and the blue dashed edges are the edges of the completion sets.}
    \label{examplechordal}
\end{figure}

\begin{lemma}
Let \( G = (V,E) \) be any chordal graph and let \( E' \) be any optimal completion set of \( G \). For any pair of distinct trees \( T_i \) and \( T_j \) of \( G \) with \( |V(T_i)| \geq 3 \), \( |V(T_j)| \geq 3 \), and $|\edgesbetweentitj|>0$, we can transform $E'$ into another optimal completion set $E''$ with $|\edgesbetweentitjprime|=0$ by replacing \( \edgesbetweentitj \) with edges within \( V(T_i) \) and \( V(T_j) \).
    \label{lem key alg}
\end{lemma}

\begin{proof}
Let \( |\edgesbetweentitj| = d >0\). We describe a procedure to modify $E'$, without increasing its size while still keeping it a completion set, so that \( |\edgesbetweentitj| < d \). This procedure, described in \cref{procedure1}, 
replaces a subset of \( \edgesbetweentitj \) with edges within \( V(T_i) \) and \( V(T_j) \). We keep applying this procedure to $T_i$ and $T_j$ until $|\edgesbetweentitj|=0$, completing the proof.

The modification procedure progresses in iterations and updates $E'$ by adding and removing edges. 

We continue this process until \( G \cup E' \) has no unsaturated edges \( \unsaturatededge=\unsaturatededgeendpoints \in E \). By \Cref{lem triangle necessary}, we do not need to worry about unsaturated edges in \( E' \). If all edges in \( E \) are saturated, we can safely remove all unsaturated edges from $E'$ using \Cref{lem triangle necessary}, yielding a smaller completion set, which contradicts the optimality of \( E' \).

At the beginning of iteration \( \ell \), we have a forest $F_{\ell} \subseteq G$ with some useful properties, which we will explain later. In each iteration, we remove exactly one cross edge and add it to a set $\setofremovededges \subseteq \edgesbetweentitj$. At the beginning of iteration $\ell \geq 2$, the set $\setofremovededges$ contains all cross edges removed during iterations $1, \dots, \ell - 1$.

\newpage
The procedure is as follows. \\ 

\begin{procedure}
\hrule

    \textbf{Reconfiguring an Optimal Completion Set.}

\noindent \textbf{Initialization:} \( F_0 = \emptyset \) and $\setofremovededges=\emptyset$.

\begin{iter}
 Pick a cross edge $\removedcrossedge=\removedcrossedgeendpoints \in \edgesbetweentitj $ using the rules stated in (i) and (ii) and modify $E'$ as described in (iii) and (iv).
\begin{enumerate}[(i)]
    \item Suppose $G$ has an edge $(b_i, b_j)\in E$ between boundary vertices $b_i \in V(T_i)$ and $b_j \in V(T_j)$ such that $(b_i, b_j)$ is contained in a triangle in $G\cup E'$ containing an edge from $E(T_i) \cup E(T_j)$. 
Suppose $(b_i, b_j)$, an edge $\unsaturatededgeendpoints \in E(T_i) \cup E(T_j)$, and a cross edge $\removedcrossedgeendpoints \in \edgesbetweentitj$ form a triangle (i.e., a $K_3$) in $G \cup E'$. Define $\removedcrossedge = \removedcrossedgeendpoints$ and $\unsaturatededge = \unsaturatededgeendpoints$.
    \item Suppose no such edge $(b_i, b_j)\in E$ exists. Then, let $\removedcrossedge=\removedcrossedgeendpoints$ be any edge from $\edgesbetweentitj$. By \Cref{lem triangle necessary}, $\removedcrossedge$ must belong to a triangle in $G \cup E'$ containing an edge from $E(T_i) \cup E(T_j)$. Let $\unsaturatededge=\unsaturatededgeendpoints\in E(T_i) \cup E(T_j)$ be this edge.
\item 
After selecting \( \removedcrossedge \) and \( \unsaturatededge \) from (i) or (ii), remove \( \removedcrossedge \) from \( E' \). If \( \unsaturatededge \) becomes unsaturated, find a \( P_3 \) in \( T_i \) (if \( \unsaturatededge \in E(T_i) \)) or \( T_j \) (if \( \unsaturatededge \in E(T_j) \)) such that \( \unsaturatededge \) is an edge of this \( P_3 \). Convert this \( P_3 \) into a \( K_3 \) by adding the missing edge \( \addededge \), i.e., update \( E' \) as \( E' \gets (E' \setminus \{\removedcrossedge\}) \cup \{\addededge\} \).

\item Set \( F_1 \xleftarrow{} (F_0 \cup \unsaturatededge)\cup \removedcrossedgeendpointtwo \) (where $\removedcrossedge=\removedcrossedgeendpoints=(v,\removedcrossedgeendpointtwo)$), $\setofremovededges \xleftarrow{}\setofremovededges \cup\{\removedcrossedge\}$, and terminate the first iteration. 
\end{enumerate}

    \label{remarkrules}
\end{iter}

\begin{iterationell}[$\mathbf{\ell\geq 2}$]
At the beginning of Iteration $\ell$, let $F_{\ell-1}$ be the forest from the previous iteration. 
\begin{enumerate}[(i)]
    \item Suppose $G\cup E'$ has an unsaturated edge $\unsaturatededge\in E$, where $\unsaturatededge=\unsaturatededgeendpoints \in E(T_i) \cup E(T_j)$  
    with $u\in V(F_{\ell-1})$ and $v\notin V(F_{\ell-1})$. 
    \item Find an edge
    $\removedcrossedge=\removedcrossedgeendpoints \in \edgesbetweentitj$ such that $\removedcrossedgeendpointone=v$ and $\removedcrossedgeendpointtwo\in V(F_{\ell-1})$.
    \item Find a $P_3$ in $T_i$ (if $\unsaturatededge \in E(T_i)$) or $T_j$ (if $\unsaturatededge \in E(T_j)$) such that $e$ is an edge of this $P_3$. Set $E'\xleftarrow{}(E' \setminus \{\removedcrossedge\})\cup \{\addededge\} $ where $\addededge$ is the missing edge of this $P_3$ (convert this $P_3$ into a $K_3$).
 
    \item Set $F_{\ell}\xleftarrow{}F_{\ell-1} \cup \unsaturatededge$ and $\setofremovededges \xleftarrow{}\setofremovededges \cup\{\removedcrossedge\}$.
\end{enumerate}
\label{iterell}    
\end{iterationell}
\noindent Terminate the procedure when $G\cup E'$ has no unsaturated edge $\unsaturatededge\in E$.
\hrule
\label{procedure1}
\end{procedure}

As an example of \Cref{remarkrules}, let $G$ and $E'$ be as depicted in \Cref{examplechordal}(a). Since $G\cup E'$ has an edge $(b_1, b_2)$ with the conditions of \Cref{remarkrules}(i), we can set $\removedcrossedge\xleftarrow{}(w_1, b_2)$, $\unsaturatededge=\unsaturatededgeendpoints\xleftarrow{} (b_1,w_1)$, $F_1\xleftarrow{}\{(b_1,w_1), b_2\}$, $\setofremovededges=\{\removedcrossedge\}=\{(w_1,b_2)\}$. Observe how the operations of \Cref{remarkrules} do not increase the size of $E'$. In the example of \cref{examplechordal}, $\removedcrossedge=(w_1,b_2)$ is removed from $E'$ (\Cref{examplechordal}(a)), and replaced with $\addededge=(z,b_1)$ (\Cref{examplechordal}(b)). An example of \Cref{procedure1} is depicted in \Cref{chordalfig3}.

We now make some observations on this procedure. First, observe that at the beginning of each iteration $\ell \geq 2$, if $G\cup E'$ has an unsaturated edge $\unsaturatededge=\unsaturatededgeendpoints$, then it must have become unsaturated due to removing edges $\setofremovededges$ in iterations $1$ to $\ell-1$, thus $\unsaturatededge \in E(T_i) \cup E(T_j)$. Moreover, $\unsaturatededge \notin E(F_{\ell-1})$ since we explicitly cover each edge in $E(F_{\ell-1})$ by adding edges $e'$. 
Therefore, by \cref{niceobsboundary} and \cref{remarkrules}(i), if $\unsaturatededge=\unsaturatededgeendpoints$ is unsaturated at the beginning of iteration $\ell$, it must have the properties described in \Cref{iterell}(i) and there must exist an edge $\removedcrossedge \in \edgesbetweentitj$ as detailed in \cref{iterell}(ii). Thus, the forest $F_{\ell}$ maintains some invariants described in the following observation.

\begin{observation}
        At the end of each iteration $\ell \geq 1$, $F_{\ell}$ has the following properties.
    \begin{enumerate}[(i)]
        \item $F_{\ell}\subseteq T_i \cup T_j$ and $F_{\ell}$ has exactly two components, $F_{\ell} \cap T_i$ and $F_{\ell}\cap T_j$.
        \item The set $\setofremovededges$ of cross edges removed in iterations $1,\dots,\ell$ lie between the two components of $F_{\ell}$ as stated in (i).
    \end{enumerate}
    \label{remarkproperty}
\end{observation}
Since we remove exactly one edge in each iteration and add at most one, the size of $E'$ never increases. The procedure terminates when $G\cup E'$ has no unsaturated edges left, at which point $|\edgesbetweentitj|<d$.

\end{proof}
  \begin{figure}[h]
    \centering
        {\includegraphics[scale=0.80]{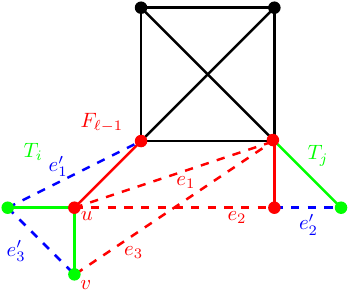}}
    \caption{An example of \Cref{iterell} for $\ell=3$: Solid red edges are the edges of $F_{\ell-1}$. Green edges belong to $T_j$ and $T_i$ but not to $F_{\ell-1}$. We remove the red dashed edges and replace them with the blue dashed ones. In this example, $\setofremovededges=\{e_1, e_2, e_3\}$ at the end of iteration $\ell$. Moreover, we add $e'_1$, $e'_2$, and $e'_3$ in the first three iterations after removing $e_1$, $e_2$, and $e_3$, respectively. Observe how $F_{\ell-1} \cup (u,v)$ remains a forest with exactly two components.} 
\label{chordalfig3}
\end{figure}
Recall the definition of outer vertices. To handle Case~II and Case~III, we can show the following lemma analogously to \Cref{lem key alg}. We omit some of the details to avoid duplication.
\begin{lemma}

Let $G$ be any connected chordal graph on at least three vertices. Then, we can modify any optimal completion set $E'$ of $G$ into another optimal completion set  $E''$ such that after adding $E''$ to $G$
\begin{enumerate}[(i)]
    \item For any pair of distinct trees $T_i$ and $T_j$ of $G$ with $|V(T_i)|\geq 3$ and $|V(T_j)|=2$, no triangles containing the edges of $E(T_i)$ contain a vertex from $V(T_j)$. Furthermore, every tree $T_j$ with $|V(T_j)|=2$ and $E(T_j) =\{\unsaturatededge=(u,v)\}$ is covered by one edge, between $u$ (resp. $v$) and a distance-two neighbour of $u$ (resp. $v$) in $G$ (as described in Line 5 of \cref{chordalalgorithm}).
    \item the endpoints of no edge in $E''$ lie between an outer vertex $u$ of $G$ and a vertex $v\in V(T_i)$ with $|V(T_i)|\geq 3$.
\end{enumerate}

    \label{claim mohem 2}
\end{lemma}
\begin{proof}
For convenience, we abuse \Cref{treesofchordaldefinition} and assume that each outer vertex of $G$ is a tree on one vertex. 

We first show \Cref{claim mohem 2}(i). 
    Let $T_i$ and $T_j$ be any two distinct trees of $G$. Using \Cref{procedure1}, we can assert that $|\edgesbetweentitj|\leq2$ and $|\edgesbetweentitj|=1$ only if $|V(T_j)|=2$ and $\removedcrossedge\in \edgesbetweentitj$ is the unique edge in $E'$ covering $\unsaturatededge_j\in E(T_j) $ in a triangle (exactly as described in Line 5 in \cref{chordalalgorithm}). 
Furthermore, $|\edgesbetweentitj|= 2$ only if $|V(T_i)|=|V(T_j)|=2$ and distinct edges $\removedcrossedge\in \edgesbetweentitj$ and $\removedcrossedge'\in \edgesbetweentitj$ are the unique edges covering $\unsaturatededge_i \in E(T_i)$ and $\unsaturatededge_j \in E(T_j)$ in triangles, respectively (exactly as described in Line 5 in \Cref{chordalalgorithm}). 
Suppose $T_i$ and $T_j$ violate the conditions mentioned above. We keep applying \Cref{procedure1} to $T_i$ and $T_j$ until the conditions are satisfied. If in some iteration $\ell$ the unsaturated edge $\unsaturatededge=\unsaturatededgeendpoints$ (see \cref{iterell}(i)) is in $E(T_j)$ with $|V(T_j)|=2$, then we ensure that the newly-added edge $\addededge$ is between $u$ (resp. $v$) and a distance-two neighbour $w$ of $u$ (resp. $v$) such that $w \notin V(T_i)$. Such a vertex $w$ always exists.

We now show \cref{claim mohem 2}(ii). 
If there exist a tree $T_i$ and an outer vertex $u$ that violate \cref{claim mohem 2}(ii), we let $T_j=u$ and keep applying \Cref{procedure1} to $T_i$ and $T_j$ until $|\edgesbetweentitj|=0$.

\end{proof}

\; \\
We summarize our result in the following.
\begin{thm}
    Let $G=(V,E)$ be a connected chordal graph on at least three vertices. \Cref{chordalalgorithm} produces an optimal completion set for $G$ in  $\bigoh(n+m)$ time, where $n=|V|$ and $m=|E|$.
    \label{thmchordal}
\end{thm}
\begin{proof}

Using \Cref{lem key alg} and \Cref{claim mohem 2}, we can reduce the $(3,1)$-cover problem on $G$ to solving the problem locally for each tree of $G$. Therefore, \Cref{chordalalgorithm} is optimal.  

As for the running time of \Cref{chordalalgorithm}, all trees of $G$ can be located by first identifying the non-bridges in $G$ in $\bigoh(n+m)$ time (using Tarjan's algorithm~\cite{tarjanbridge}) and removing them from $G$. The non-trivial components of the resulting graph correspond to the trees of $G$. For each tree $T_i$ of $G$, we can compute its optimal $(3,1)$-cover in $\bigoh(|V(T_i)|)$ time by \Cref{proptreeoptimal}. Thus, in $\bigoh(n+m)$ time, we can construct an optimal $(3,1)$-cover of $G$.
\end{proof}

\section{The $(k,1)$-Cover and the $(3,k-2)$-Cover Problem for Trees}
\label{sec4}
In this section, we present constant-factor approximation algorithms for the $(k,1)$-cover and the $(3,k-2)$-cover problems for $k\geq 5$ (\Cref{sec41}) and $k=4$ (\cref{sec42}).
Since every graph with a $(k,1)$-cover (with $k\geq 3$) trivially has a $(3,k-2)$-cover, the lower bound of \eqref{optlowerbound} also holds for any $k$-completion set. Therefore, in the remainder of this paper, we use this lower bound to prove our approximation ratios for trees.
\subsection{{An Approximation Algorithm for $k\geq 5$}}
\label{sec41}
The algorithm is presented in \cref{twoapproxprime}. In each iteration, \cref{twoapproxprime} extracts a maximal $k$-sub-forest from $T$, turns it into a $k$-clique by adding edges, and removes the edges of this sub-forest from $T$. It then repeats this procedure on the remaining forest until no edges are left. We say a forest is a \defin{$k$-forest} if it has at most $k$ vertices and no singleton components, i.e., all of its components are non-trivial. Let $G$ be any forest. We say $k$-forest $H$ is a \defin{maximal $k$-sub-forest} of $G$ if $H\subseteq G$ and there exists no $k$-forest $H'\neq H$ such that $H'\subseteq G$ and $H\subseteq H'$. We first make a remark on Line 5 of \Cref{twoapproxprime}.
\begin{remark}
    The maximal sub-forest $F_j$ mentioned in Line 5 of \cref{twoapproxprime} can be found in the following way. Initially, $F_j \xleftarrow[]{}\emptyset$. If $T \sm (E_0 \cup \dots \cup E_{j-1})$ has a tree $T_i$ with $|V(T_i)|\geq k$, then using a simple traversal, we can find a sub-tree of $T_i$ with exactly $k$ vertices (see \Cref{example1}). If for every tree $T_i$ of $T \sm (E_0 \cup \dots \cup E_{j-1})$ we have $|V(T_i)|<k$, then we set $F_j \xleftarrow{} F_j \cup T_i$ for some tree $T_i$ of $T \sm (E_0 \cup \dots \cup E_{j-1})$ with $|V(T_i)|>1$. We then recurse on $(T\sm (E_0 \cup \dots \cup E_{j-1}))\sm T_i$ and search for a maximal ($k -|V(T_i)|$)-sub-forest 
    of $(T\sm (E_0 \cup \dots \cup E_{j-1}))\sm T_i$ (see \Cref{example2}). This process is continued until $F_j$ cannot be extended further, i.e., it is maximal.
    \label{remarkalg}
\end{remark}
\begin{example}
    An example of the procedure described in \Cref{remarkalg} is depicted in \Cref{example1fig} for $k=7$. In \Cref{example1fig}, $T \sm (E_0 \cup \dots \cup E_{j-1})=T_1 \cup T_2$ (we ignore the singleton components), and $(T \sm (E_0 \cup \dots \cup E_{j-1})) \cap F_j$ is depicted in red. Since $T \sm (E_0 \cup \dots \cup E_{j-1})$ has a tree $T_1$ with $|V(T_1)|\geq 7$, then $F_j$ is set to be a sub-tree of $T_1$ on seven vertices. This sub-tree of $T_1$ can be found by applying any traversal algorithm (e.g., BFS) to $T_1$. 
    \label{example1}
\end{example}
  \begin{figure}[tbh]
    \centering
        {\includegraphics[scale=0.80]{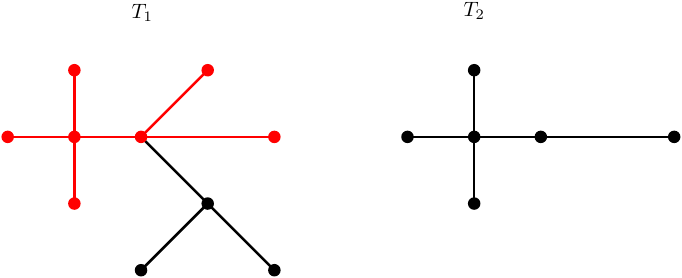}}
    \caption{An example of the procedure of \Cref{remarkalg} as described in \Cref{example1}. $T \sm (E_0 \cup \dots \cup E_{j-1})=T_1 \cup T_2$, and $(T \sm (E_0 \cup \dots \cup E_{j-1})) \cap F_j$ is depicted in red.} 
    \label{example1fig}
\end{figure}
\begin{example}
    Another example of the procedure of \Cref{remarkalg} for $k=7$ is depicted in \Cref{example2fig}. Similar to \Cref{example1}, $T \sm (E_0 \cup \dots \cup E_{j-1})=T_1 \cup T_2$ (ignoring the singleton components), and $(T \sm (E_0 \cup \dots \cup E_{j-1})) \cap F_j$ is depicted in red. Since we have $|V(T_1)|=|V(T_2)|<7$, in the first step, the procedure of \Cref{remarkalg} sets $F_j \xleftarrow[]{} F_j \cup T_1$. Then, the procedure looks for a forest with at most $k-|V(T_1)|=7-5=2$ vertices from $(T \sm (E_0 \cup \dots \cup E_{j-1}))\setminus T_1=T_2$. This is done by traversing $T_2$ and extracting a sub-tree on two vertices. 
    \label{example2} 
\end{example}
  \begin{figure}[tbh]
    \centering
        {\includegraphics[scale=0.75]{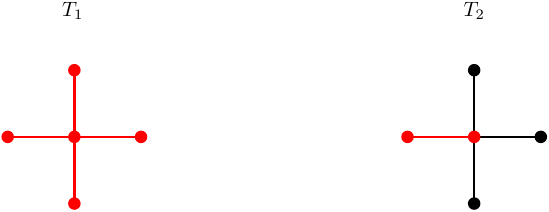}}
    \caption{An example of the procedure of \Cref{remarkalg} as described in \Cref{example2}. $T \sm (E_0 \cup \dots \cup E_{j-1})=T_1 \cup T_2$, and $(T \sm (E_0 \cup \dots \cup E_{j-1})) \cap F_j$ is depicted in red.} 
    \label{example2fig}
\end{figure}

\begin{algorithm}[H]
    \caption{The approximation algorithm of \cref{sec41}}
    \label{twoapproxprime}
    \begin{algorithmic}[1]
    \State \textbf{Input:} A tree $T=(V,E)$ and an integer $k \geq 5$ with $k \leq |V|$
    \State \textbf{Output:} A $k$-completion set $S$ of $T$.
    \State \textbf{Initialization:} $S \xleftarrow{} \emptyset$, $j \xleftarrow{} 1$, $E_0 \xleftarrow{} \emptyset$
   \While{$E \sm (E_0 \cup \dots \cup E_{j-1}) \neq \emptyset$}
   \State Find a maximal $k$-sub-forest with no singleton components (where every component is non-trivial) of $T \sm (E_0 \cup \dots \cup E_{j-1})$. Let  $F_j=(V_j,E_j)$ denote this sub-forest (see \Cref{remarkalg}). 

   \State Turn $F_j$ into a $k$-clique by adding edges and update $S$ accordingly. If $|V_j|<k$, let $V'_j$ be $k-|V_j|$ arbitrary vertices from $V \sm V_j$. Set $F_j \xleftarrow[]{} F_j\cup V'_j$, turn $F_j$ into a $k$-clique by adding edges, and update $S$ accordingly.
   \State $j \xleftarrow{}j+1$
   \EndWhile
  \State \Return $S$
    \end{algorithmic}
    \end{algorithm}
    \;
    We now show the correctness of \cref{twoapproxprime} in the following theorem.
    \begin{thm}
        Let $k\geq 5$ be any integer. \cref{twoapproxprime} is an $(\frac{8}{3})$-approximation for the $(k,1)$-cover and the $(3,k-2)$-cover problems when the input graph is any tree on at least $k$ vertices. 
        
        Furthermore, when the input graph is an $n$-vertex tree, \cref{twoapproxprime} runs in $\bigoh(n k)$ time.
        \label{thmtwoapprox}
    \end{thm}
    \begin{proof}
In this proof, we assume that the input graph is a tree $T=(V,E)$ with $|V|=n$. It is easy to see that for the returned set $S$ in Line 9, $T \cup S$ must necessarily have a $(k,1)$-cover. Therefore, in the remainder of the proof, we focus on proving the approximation ratio. 

Since the sub-forest $F_j$ in Line 5 of \cref{twoapproxprime} is a maximal one whose every component has at least two vertices, the following claim is easy to show.
\begin{claim}
Let $F_j$ be the maximal $k$-sub-forest in Line 5 of \cref{twoapproxprime} for any iteration $j$. The following statements are true.
\begin{enumerate}[(a)]
    \item If $k$ is even and $|E(T \sm (E_0 \cup \dots \cup E_{j-1}))|\geq \frac{k}{2}$, then $|E_j|\geq \frac{k}{2}$.
    \item  If $k$ is odd and  $|E(T \sm (E_0 \cup \dots \cup E_{j-1}))|\geq \frac{k-1}{2}$, then $|E_j|\geq \frac{k-1}{2}$.
\end{enumerate}
    \label{claimsizeej}
\end{claim}

We consider two cases for the approximation ratio.\\

\noindent\textbf{Case I:} {$k$ is even.} Suppose $|E|=n-1\geq2k$, we will return to the case with $|E|<2k$ later. Let $\alg$ denote the number of edges added to $T$ by \cref{twoapproxprime}, and let $\iteration$ denote the total number of iterations of the main loop of \cref{twoapproxprime}. During each iteration $j$, \cref{twoapproxprime} adds at most $\binom{k}{2}-|E_j|$ edges to $S$. We have

\begin{align}
    \alg &\leq \sum_{j=1}^{\iteration}\binom{k}{2}-|E_j|\nonumber\\
    &=\iteration \times \binom{k}{2}-\sum_{j=1}^{\iteration}|E_j|=\iteration \times \binom{k}{2}-|E|=\iteration\times \binom{k}{2}-(n-1) .\label{eq2}
\end{align}
Therefore, $\alg$ is maximized when $\iteration$ is maximized. Let $|E|=(n-1)= q\times (\frac{k}{2})+r$ with $0\leq r<\frac{k}{2}$ and $0\leq q$. Using \cref{claimsizeej}(a), for every iteration $j<\iteration$ we have $|E_j|\geq \frac{k}{2}$. Since at every iteration we process a subset of $E$ and the loop terminates when all edges are processed, then it is easy to see that $\iteration \leq q+1$ because at every iteration we process at least $\frac{k}{2}$ edges. Furthermore, it is easy to see that $\iteration \leq q$ when $r=0$ and $\iteration\leq q+1$ when $0<r<\frac{k}{2}$. From \eqref{optlowerbound} and $n-1=q\times \frac{k}{2}+r$ we get

\begin{equation}
    \opt\geq (n-1) \times \bigg(\frac{k-2}{2}\bigg)=\frac{q\times k\times (k-2)}{4}+r\times \frac{(k-2)}{2}
    \label{eq3}
\end{equation}
when $r=0$, $\iteration\leq q$ and \eqref{eq2} is maximized when $|E_j|=\frac{k}{2}$ for $q$ many iterations with 
$$\alg \leq q\times \bigg(\binom{k}{2}-\frac{k}{2}\bigg)=\frac{q\times k \times (k-2)}{2}= 2\times \opt \text{ (see \eqref{eq3})}$$ 
Now suppose $0<r<\frac{k}{2} $. We have $\iteration\leq q+1$ and \eqref{eq2} is maximized when $|E_j|=\frac{k}{2}$ for $q$ many iterations and $|E_j|=r$ for one iteration. Therefore, we have

\begin{align}
\alg &\leq q\times \bigg(\binom{k}{2}-\frac{k}{2}\bigg) + \binom{k}{2}-r=\frac{q\times k \times (k-2)}{2}+ \frac{k\times(k-1)}{2}-r \label{firsteq1}\\
&\leq \frac{2\times q\times k\times(k-2) +2\times k\times(k-1) -4\times r}{q\times k \times (k-2)+2\times r \times (k-2)}\times \opt \label{eq4}\\
&\leq \bigg(2+\frac{2\times(k-1)\times(k-2 \times r)}{(k-2)\times(k \times q+2 \times r)}\bigg)\times \opt \label{eq5}\\
&\leq  \bigg(2+\frac{2\times(k-1)}{(k-2)\times q}\bigg)\times \opt \label{eq5prime} \\
&\leq \bigg(2.625\bigg)\times \opt \label{eq5primeprime}
    \end{align}
      where \eqref{firsteq1} to \eqref{eq4} holds due to \eqref{eq3}. Moreover, \eqref{eq4} to \eqref{eq5} and \eqref{eq5} to \eqref{eq5prime} hold for any $q>0$, $k>5$, and $0< r < \frac{k}{2}$.  Since $n-1 \geq 2k$ by assumption, we have $q\geq 4$. Therefore, \eqref{eq5prime} to \eqref{eq5primeprime} holds for any $k\geq 6$ and $q\geq 4$. 
      
      {Now when $n-1< 2k$, note that the procedure described in Line 5 of \cref{twoapproxprime} for extracting maximal forests always returns a $k$-vertex tree in the first iteration, i.e., $|E_1|=k-1$. Therefore, when $n-1=|E|\leq 2\times k$, \cref{twoapproxprime} terminates in at most two iterations when $2\leq q< 3$ (with $\alg \leq 2\binom{k}{2}-(n-1)$), and in at most three iterations when $3\leq q<4$ (with $\alg\leq 3\binom{k}{2}-(n-1)$). In any case, it can be shown that $\alg\leq 2\opt$. We omit the proof of this last claim to avoid duplication. }

This concludes the proof for the case when $k$ is even.
      
    \noindent \textbf{Case II:} $k$ is odd.
  
We write $|E|$ as $|E|=(n-1)= q\times (\frac{k-1}{2})+r$ with $0\leq r<\frac{k-1}{2}$ and $0\leq q$. Similar to \eqref{eq3}, we have
\begin{equation}
    \opt \geq  \bigg(\frac{k-2}{2}\bigg)\times (n-1)=q\times \frac{(k-1)\times(k-2)}{4}+ r\times (\frac{k-2}{2})  
    \label{optodd}
\end{equation}

We tighten our analysis for this case. Note that the procedure described in Line 5 of \cref{twoapproxprime} for extracting maximal forests always returns a $k$-vertex tree in the first iteration, i.e., $|E_1|=k-1$. Therefore, $T\sm E(T_1)$ has exactly $q\times(\frac{k-1}{2})+r-(k-1)=(q-2)\times (\frac{k-1}{2})+r$ with $0\leq r<\frac{k-1}{2}$ and $q-2\geq 0$. Using \cref{claimsizeej}(b), \cref{twoapproxprime} runs for at most $q-2 +1$ iterations after the first iteration and in total $\iteration\leq q$. By setting $\iteration=q$ in \eqref{eq2} we get
\begin{align}
\alg \leq q\times \binom{k}{2} -(n-1)&=\frac{q\times k\times (k-1)}{2}-\bigg(q\times \frac{k-1}{2}+r\bigg)\nonumber\\
&\leq q\times \frac{(k-1)^2}{2} \label{eq7}
    \end{align}
Using \eqref{optodd} and \eqref{eq7} we get
\begin{align}
\alg &\leq \frac{2\times q\times (k-1)^2}{q\times(k-1)\times (k-2)+2\times r\times(k-2)}\times \opt\nonumber\\
&\leq \frac{2\times q\times (k-1)^2}{q\times(k-1)\times (k-2)}\times \opt \nonumber\\
&= \frac{2\times (k-1)}{(k-2)}\times \opt \nonumber \nonumber\\
&\leq \frac{8}{3}\times \opt \nonumber \nonumber
    \end{align}
where the last inequality holds for any $k\geq 5$.

To finalize the proof, we show that \cref{twoapproxprime} runs in $\bigoh(n k)$ time. From \cref{claimsizeej}, we know that there are at most $\bigoh(\frac{n}{k})$ iterations. During each iteration $j$, finding the maximal subgraph $F_j$ and building a $k$-clique  on $F_j$ can each be done in $\mathcal{O}(k^2)$ time. Therefore, the time complexity of \cref{twoapproxprime} is $\bigoh(\frac{n}{k}  k^2)=\bigoh(n k)$.
\end{proof}

\subsection{{A Better Approximation Algorithm for $k=4$}}
\label{sec42}
In this section, we present a $2$-approximation Algorithm for the $(4,1)$-cover and $(3,2)$-cover problems for trees. As seen in \cref{sec41}, {greedily} extracting sub-forests from the input tree $T$ results in an {$(\frac{8}{3})$-approximation} algorithm for the $(k,1)$-cover problem. In this section, we show that cutting each subgraph of $T$ more carefully improves this approximation ratio for the case of $(4,1)$-cover and $(3,2)$-cover problems. 

We now briefly describe our algorithm, which is presented in \cref{givefourapproxprime}. \cref{givefourapproxprime} extracts a 4-vertex sub-tree $T_j$ of $T$ at every step of the main loop (Line 5), turns $T_j$ into a 4-clique, and removes the edges of $T_j$ from $T$. If the biggest non-trivial component $T^{(j)}_{\operatorname{max}}$ of $T \sm E(T_j)$ has at least four vertices, \cref{givefourapproxprime} sets $T\xleftarrow{} T^{(j)}_{\operatorname{max}}$. Otherwise, it sets $T\xleftarrow{} \emptyset$ which terminates the main loop. Since each iteration of the main loop of \cref{givefourapproxprime} only recurses on the biggest non-trivial component of $T\sm E(T_j)$, the edges that are ignored by the main loop get stored in a forest $F$ (Line 24 and Line 26). When the main loop terminates, these edges are handled by applying \cref{twoapproxprime} to $F$ in Line 30 of \cref{givefourapproxprime}.

\scalebox{0.79}{ 
\begin{minipage}{1.1\textwidth} 
    \begin{algorithm}[H]
    \caption{The approximation algorithm of \cref{sec42}}
    \label{givefourapproxprime}
    \begin{algorithmic}[1]
    \State \textbf{Input:} A tree $T=(V,E)$ rooted at a vertex $r$ with $|V|\geq 4$
    \State \textbf{Output:} A $4$-completion set $S$ of $T$
    \State \textbf{Initialization:}
     $S \xleftarrow{} \emptyset$, $j \xleftarrow{} 1$, $F \xleftarrow{} \emptyset$ 
   \While{$E(T) \neq \emptyset$}
   \State Let $v_j$ be the deepest leaf in $T$ with the greatest number of siblings. Let $u$ be the parent of $v_j$. Find a 4-vertex sub-tree $T_j=(V_j,E_j)$ of $T$ in the following way. 
   \If{$v_j$ has at least two siblings}
   \State Let $v_1$ and $v_2$ be two such siblings. Set $T_j\xleftarrow{} T[\{v_j,v_1,v_2,u\}]$
   \EndIf
      \If{$v_j$ has exactly one sibling}
         \State Let $v_1$ be this sibling and let $w$ be the parent of $u$. Set $T_j\xleftarrow{} T[\{v_j,v_1,u,w\}]$
   \EndIf

    \If{$v_j$ has no siblings}
          \State Let $u,w, x$ be the immediate ancestors of $v_j$ (in the same order) 
    \If{$u$ has a sibling $u_1$}
    \State  Set $T_j\xleftarrow{} T[\{v_j,u,u_1,w\}]$
    \EndIf

        \If{$u$ has no siblings}
    \State  Set $T_j\xleftarrow{} T[\{v_j,u,w,x\}]$
    \EndIf
   \EndIf
   
   \State Turn $T_j$ into a $4$-clique by adding edges. Update $S$ accordingly. 
   \State Let $T^{(j)}_{\operatorname{max}}$ be the biggest non-trivial component of $T\sm E(T_j)$. If $T\sm E(T_j)$ has no non-trivial components, set $T^{(j)}_{\operatorname{max}} \xleftarrow{}\emptyset$.
   \If{$|V(T^{(j)}_{\operatorname{max}})|\geq 4$}
   \State $F\xleftarrow{} F \cup (E(T)\sm (E(T_j)\cup E(T^{(j)}_{\operatorname{max}})))$ and $T \xleftarrow[]{} T^{(j)}_{\operatorname{max}}$ 
   \Else
   \State $F\xleftarrow{}F\cup ( E(T) \sm E(T_j))$ and $T\xleftarrow{}\emptyset$ 
   \EndIf
   \State $j \xleftarrow{}j+1$
   \EndWhile
   \State Run the main loop of \cref{twoapproxprime} on $F$ with $k=4$ and update $S$ accordingly.   
  \State \Return $S$
    \end{algorithmic}
    \end{algorithm}
\end{minipage}}
    
     Before proving the correctness, we make one final note on Line 5 of \cref{givefourapproxprime}.
We assume that the input tree $T$ is rooted at an arbitrary node $r$, and each node except $r$ has a parent. In Line 5, we first find the deepest leaf $v_j$ of the current tree $T$ or the furthest node from $r$ in $T$. If there are two or more such leaves, we pick the one with the most siblings. The tree $T_j$ is obtained by the conditions in Lines 6 to 20 of \cref{givefourapproxprime}. The next lemma is essential for proving the approximation ratio.
    \begin{lemma}
After selecting $T_j$ in any iteration $j$ of \cref{givefourapproxprime}, $T \sm E(T_j)$ has at most two non-trivial components. Furthermore, if $T \sm E(T_j)$ has exactly two non-trivial components, it has at least one component isomorphic to $K_2$.
\label{lemmakeygivefour}
    \end{lemma} 
    \begin{proof}
        Let $v_j$ be the chosen leaf in Line 5 of iteration $j$. There are only a few cases to consider as described in Lines 6 to 20 of \cref{givefourapproxprime}.  If the two immediate ancestors of $v_j$ are of degree two in $T$, then $T_j$ is a path on four vertices containing $v_j$ and its ancestors (\cref{lem5fig}(a)) and $T \sm E(T_j)$ has at most one non-trivial component $T^{(j)}_{\operatorname{max}}$ (this corresponds to Lines 17 to 19 of \cref{givefourapproxprime}). If $v_j$ has at least two other siblings in $T$, then $T_j$ is isomorphic to $K_{1,3}$ where the middle vertex is the parent of $v_j$ and the other vertices are the siblings of $v_j$ (\cref{lem5fig}(b)) and $T \sm E(T_j)$ has at most one non-trivial component $T^{(j)}_{\operatorname{max}}$ (this corresponds to Lines 6 to 8 of \cref{givefourapproxprime}). Similarly, if $v_j$ has only one sibling, $T_j$ will be isomorphic to $K_{1,3}$ with $v_j$'s parent in the middle and two other vertices including the other sibling of $v_j$ and the second ancestor of $v_j$ (\cref{lem5fig}(c)) (this corresponds to Lines 9 to 11 of \cref{givefourapproxprime}). 

        We now focus on Lines 14 to 16 of \cref{givefourapproxprime}. Let $w$ denote the second ancestor of $v_j$ in $T$. The cases where $v_j$ has no siblings but $u$ has at least one sibling are depicted in \cref{lem5fig}(d) and \cref{lem5fig}(e). Let $u_1$ denote the other child of $w$ (other than $u$). Since $v_j$ is the deepest leaf of $T$, the sub-tree rooted at $u_1$ is of depth at most one. Furthermore, since $v_j$ is the deepest leaf with the most siblings (Line 5 of \cref{givefourapproxprime}), $u_1$ can have at most one child. If $u_1$ has no children, then $T_j$ is isomorphic to a path on three vertices containing $v_j$, $u$, $w$ and $u_1$ and $T \sm E(T_j)$ has at most one non-trivial component $T^{(j)}_{\operatorname{max}}$ (\cref{lem5fig}(d)). Similarly, if $u_1$ has exactly one child, then $T_j$ will have a similar structure and $T \sm E(T_j)$ will have at most two non-trivial components $T^{(j)}_{\operatorname{max}}$ and $T'_j$, where $T'_j$ is isomorphic to $K_2$ (\cref{lem5fig}(e)).
    \end{proof} 
    
              \begin{figure}[H]
    \centering
        \subfloat[]{{\includegraphics[scale=0.60]{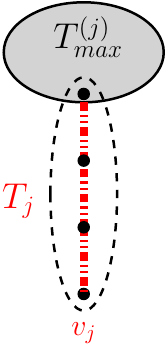}}}
        \qquad
        \subfloat[]{{\includegraphics[scale=0.60]{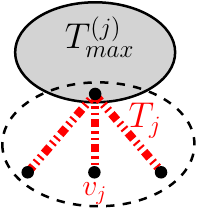}}}
        \qquad
         \subfloat[]{{\includegraphics[scale=0.60]{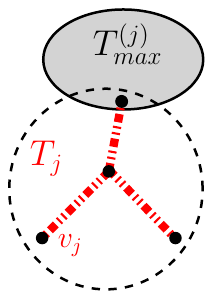}}}\qquad
          \subfloat[]{{\includegraphics[scale=0.60]{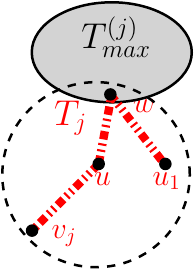}}}
          \qquad
        \subfloat[]{{\includegraphics[scale=0.60]{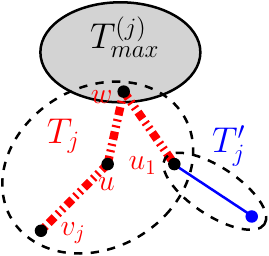}}}
    \caption{An illustration of the proof of \cref{lemmakeygivefour}.  }
    \label{lem5fig}
\end{figure}
In the next lemma, we characterize the forest $F$ right after the loop of Lines 4 to 29 of \cref{givefourapproxprime}. Recall that $P_k$ denotes a path on $k$ vertices.

\begin{lemma}
Once the main loop of Lines 4 to 29 of \cref{givefourapproxprime} terminates, the non-trivial components of $F$ are isomorphic to a disjoint union of $K_2$'s, with at most one $P_3$. 
    \label{matchingforest}
\end{lemma}
\begin{proof}
   From \cref{lemmakeygivefour} recall that each iteration of the main loop of \cref{givefourapproxprime} leaves at most one non-trivial component uncovered, as depicted in \cref{lem5fig}(e). Therefore, for every iteration except the last one, the non-trivial components of $F$ are isomorphic to a disjoint union of $K_2$'s. Furthermore, since the main loop continues until $|V(T)|\geq 4$, if the else condition in Line 25 of the last iteration is satisfied, then $F$ will have another $K_2$ or a $P_3$.
\end{proof}
We are now ready to prove the correctness of \cref{givefourapproxprime}.
\begin{thm}
        \cref{givefourapproxprime} is a 2-approximation algorithm for the $(4,1)$-cover and the $(3,2)$-cover problems when the input is an $n$-vertex tree $T=(V,E)$. Moreover, \cref{givefourapproxprime} terminates in $\bigoh(n\log n)$ time. 
    \label{thmgivefour}
\end{thm}
\begin{proof}
    Note that once \cref{givefourapproxprime} terminates, all edges of $T \cup S$ are in $4$-cliques (and trivially at least two triangles). 
    
    We start by proving the time complexity. The deepest leaf with the most siblings (Line 5 in \Cref{givefourapproxprime}) can be found using the following data structure. We maintain an array \( A \), where \( A[i] \) contains a max-heap that stores the non-leaf nodes at depth \( i \). The nodes in each heap are organized based on the number of children they have. In the first iteration, we locate the last non-empty cell of \( A \) and extract a node \( u \) from its heap. Any child \( v \) of \( u \) is the deepest leaf with the maximum number of siblings. After extracting \( T_j \), some nodes may become singletons or leaves. These nodes are deleted from their corresponding heaps. Moreover, if some non-leaf node

Note that we never change the depth of any node during the algorithm. If the number of children of a node changes, we update the corresponding heap in \( \bigoh(\log n) \) time. The overall time complexity of the algorithm is \( \bigoh(n \log n) \). This is achieved because there are at most \( n \) iterations, and each heap operation requires \( \bigoh(\log n) \) time for a constant number of heap changes per iteration.

    We now prove the approximation ratio. Similar to the proof of \cref{thmtwoapprox}, let $\opt$ denote the size of any optimal solution to the $(4,1)$-cover and $(3,2)$-cover problems, and we denote $|S|$ by $\alg$. Furthermore, let $\iteration$ denote the total number of iterations of \cref{givefourapproxprime}, that is the number of iterations of the loop of Lines 4 to 29, plus the number of iterations of \cref{twoapproxprime} as invoked in Line 30 of \cref{givefourapproxprime}. For each such iteration $j \in \{1, \dots,  \iteration\}$, we denote by $E_j$ the set of edges \textit{covered} in that iteration, see Line 5 of \cref{givefourapproxprime} and Line 5 of \cref{twoapproxprime}. In the remainder of this proof, we refer to the loop in Lines 4 to 29 of \cref{givefourapproxprime} as \textit{the first phase} of \cref{givefourapproxprime}, and the loop of \cref{twoapproxprime} as invoked on Line 30 of \cref{givefourapproxprime} as \textit{the second phase} of \cref{givefourapproxprime}. 
    
    We have the following claims.
    \begin{claim}
        For every iteration $j$ of the first phase of \cref{givefourapproxprime}, we have $|E_j|=3$. One of the iterations of the second phase may have $|E_{j}|=1$, and for every other iteration,  $|E_j|=2$.

        \label{claim4}
    \end{claim}
    \begin{proof}
Immediate from \cref{matchingforest} and \cref{twoapproxprime}.
    \end{proof}
    \begin{claim}
    Let $T$ be any tree input to \cref{givefourapproxprime} and let $\iteration$ be the total number of iterations. Then, $\iteration$ is maximized when the number of iterations of the first phase is minimized. 
        \label{claim5}
    \end{claim}
    \begin{proof}
        Let $\iteration_1$ and $\iteration_2$ denote the number of iterations of the first and second phases of \cref{givefourapproxprime}, respectively. Since in the first phase we have $|E_j|=3$ for every iteration $j$ (\cref{claim4}), \cref{givefourapproxprime} covers $3\times \iteration_1$ edges in the first phase. By the second part of \cref{claim4}, we have $\iteration_2=\bigg\lceil \frac{|E|-3\times \iteration_1}{2}\bigg\rceil$. It follows that $\iteration=\iteration_1+\iteration_2=\iteration_1 + \bigg\lceil \frac{|E|-3\times \iteration_1}{2}\bigg\rceil$ and it is easy to see that for any value of $|E|$, $\iteration$ is maximized when $\iteration_1$ is minimized.  
    \end{proof}
    The following observation states that for any $n$-vertex tree, the first phase has at least $\big \lfloor \frac{n-1}{4}\big \rfloor$ iterations. 
    \begin{observation}
   Let $\iteration_1$ and $\iteration_2$ denote the number of iterations of the first and the second phases of \cref{givefourapproxprime}, respectively. For any tree $n$-vertex tree $T$, let $n-1= q \times 4 +r$ with $0\leq r<4$ and $0 \leq q$. Then, $q\leq \iteration_1$ if $r<3$ and $q+1\leq \iteration$ if $r=3$. Furthermore, these lower bounds can be achieved if every iteration of the first phase except at most one (the last one) results in the case described in \cref{lemmakeygivefour} and depicted in \cref{lem5fig}(e).
        \label{obs5}
    \end{observation}
    We present one final claim.
    \begin{claim}
        The maximum value of $\alg$ for any $n$-vertex tree is obtained when every iteration of the first phase results in the case depicted in \cref{lem5fig}(e).
\label{claim6}
    \end{claim}
    \begin{proof}
    Let $\iteration_1$ and $\iteration_2$ denote the total number of iterations of the first and second phases, respectively. Similar to the proof of \cref{thmtwoapprox}, we can write
    \begin{align}
    \alg &\leq \sum_{j=1}^{\iteration}\binom{4}{2}-|E_j|\nonumber\\
    &=6\iteration  -\sum_{j=1}^{\iteration}|E_j|=6\iteration-|E|=6\iteration-(n-1) \label{eq2prime}
\end{align}
where $\iteration=\iteration_1+\iteration_2$. Therefore, $\alg$ is maximized when $\iteration$ is maximized. Using \cref{claim5}, $\iteration$ is maximized when $\iteration_1$ is minimized. By \cref{obs5}, the minimum value of $\iteration_1$ over all $n$-vertex trees is attained when the case in \cref{lem5fig}(e) happens at every iteration of the first phase. 
    \end{proof}
    We now prove the approximation ratio by proving an upper bound on $\alg$ when the input is any tree $T=(V,E)$. Since the worst-case scenario for $\alg$ is valid for any $n$-vertex tree, we only count the number of edges added in the situation described in \cref{claim6}. Let us write
    \begin{equation}
        n-1=4\times q +r, \; 0\leq r<4, \;0\leq q
        \label{eq9}
    \end{equation}
\begin{equation}
    q=2\times q'+r', \; 0 \leq r'<2,\; 0 \leq q'
    \label{eq10}
\end{equation}
By \eqref{eq9}, \eqref{eq10}, and \eqref{optlowerbound} we have
\begin{equation}
    \opt \geq n-1= 8\times q'+ 4\times r'+r
\end{equation}
Based on the values of $0\leq r<4 $ and $0\leq r'<2$, there are a total of eight cases to consider. Here, we only consider three interesting cases. The remaining cases can be shown analogously.

\noindent \textbf{Case I:} $r'=0$ and $r=0$. In this case, \cref{givefourapproxprime} adds $3 \times q$ edges to $S$ in the first phase (because the first phase has $q$ many iterations $j$ with $|E_j|=3$), and at most $4\times q'$ edges in the second phase (because the second phase has $q'$ many iterations $j$ with $|E_j|=2$). We have
$$\frac{\alg}{\opt} \leq \frac{3\times q +4\times q'}{8\times q'}= \frac{6\times q' +4\times q'}{8\times q'}=\frac{5}{4}$$

\noindent \textbf{Case II:} $r'=0$ and $r=1$. In this case, \cref{givefourapproxprime} adds $3 \times q$ edges to $S$ in the first phase, and at most $4\times q'+5 $ edges in the second phase. We have
$$\frac{\alg}{\opt} \leq \frac{3\times q +4\times q'+5}{8\times q'+1}= \frac{10\times q' +5}{8\times q'+1}\leq \frac{5}{4}+ \frac{15}{4\times(8 \times q'+1)}\leq \frac{5}{4}+\frac{15}{36}$$\
where the last inequality holds for any $q'\geq 1$. If $q'=0$, then from \eqref{eq10} we have $q=0$ and from \eqref{eq9} we deduce $n-1=1$, a contraction since $n-1\geq 3$.

\noindent \textbf{Case III:} $r'=0$ and $r=2$. In this case, \cref{givefourapproxprime} adds $3 \times q$ edges to $S$ in the first phase, and $4\times (q'+1)$ edges in the second phase. We have
$$\frac{\alg}{\opt} \leq \frac{3\times q +4\times q'+4}{8\times q'+2}= \frac{10\times q' +4}{8\times q'+2}\leq \frac{5}{4}+ \frac{3}{4\times(4 \times q'+1)}\leq 2$$
where the last inequality holds for $q'\geq 0$.

For every other case not described above, we always have $\frac{\alg}{\opt} \leq \frac{5}{4}+ \frac{a}{b\times (4\times q' +c)}$ for some constants $a$, $b$, and $c$, and by setting $q'=0$ it is observed that the ratio is bounded by two.
\end{proof}
\begin{remark}
     Since in all cases of the proof of \Cref{thmgivefour} the ratio is bounded by $\frac{\alg}{\opt} \leq \frac{5}{4}+ \frac{a}{b\times (4\times q' +c)}$ for $q'\in \Theta(n)$ and constants $a$, $b$, and $c$, \cref{givefourapproxprime} is a $(\frac{5}{4}+o(1))$-approximation algorithm for the $(4,1)$-cover and the $(3,2)$-cover problems. For example, when $n\geq 500$, \Cref{givefourapproxprime} is a (1.26)-approximation algorithm. 
\end{remark}
\begin{remark}
To prove the approximation ratio in \cref{thmgivefour}, we compared the worst-case scenario of \cref{givefourapproxprime} (as stated in \cref{claim6}) with the best possible graph-theoretic bounds as stated in \eqref{optlowerbound}. For every $n$, there exists an $n$-vertex tree $T$ such that inputting $T$ to \cref{givefourapproxprime} can result in the worst-case scenario of \cref{claim6}. Consider the spider graph with the central vertex $v$ as the root. If $n-1$ is even, we add $\frac{n-1}{2}$ legs to $v$, each with two edges. If $n-1$ is odd, we add $\lfloor \frac{n-1}{2}\rfloor$ legs of length two and one leg of length one to $v$. An example for $n=15$ is depicted in \cref{fig3} with each $v_j$ and $T_j$ highlighted (see Line 5 of \cref{givefourapproxprime}).

\end{remark}\label{remark4}
     \begin{figure}[H]
    \centering
        {\includegraphics[scale=0.7]{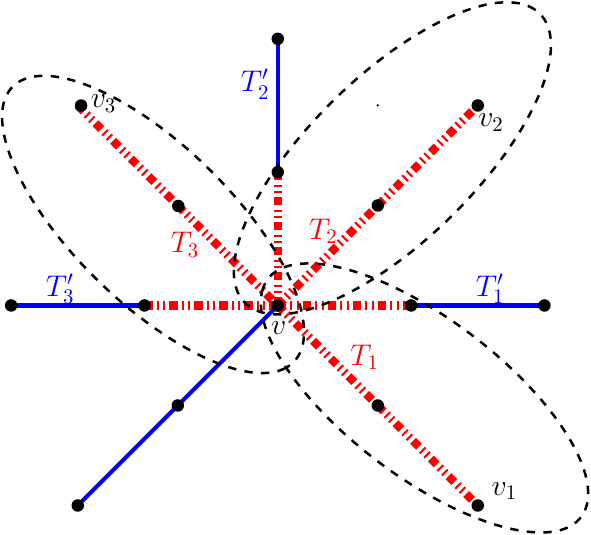}}
    \caption{ The example described in \cref{remark4}.} 
    \label{fig3}
\end{figure}
\section{Conclusion}
\label{conclude}
In this paper, we studied the algorithmic version of the $(k,1)$-cover problem. We proved that the $(k,1)$-cover problem is $\mathbb{NP}$-complete for general graphs.
However, we showed that for chordal graphs, the $(3,1)$-cover problem can be solved in polynomial time. Specifically, we provided an algorithm that runs in $\mathcal{O}(n + m)$ time, where $n$ and $m$ are the number of vertices and edges, respectively.
For the class of trees and general values of $k$, we showed that the $(k,1)$-cover problem remains $\mathbb{NP}$-hard, even for spiders. Moreover, we presented an $\frac{8}{3}$-approximation algorithm for both the $(k,1)$-cover and the $(3, k-2)$-cover problems, which runs in $\mathcal{O}(nk)$ time for trees, where $k \geq 5$ and for $k=4$, a 2-approximation algorithm that runs in $\bigoh(n\log n)$ time.

We close the paper with some open questions.

\begin{itemize}
    \item     Is there an approximation algorithm for the $(k,1)$-cover problem on general graphs with a non-trivial approximation ratio? 
    \label{problem1}
    \item   Can the approximation ratios of \Cref{thmtwoapprox} and \Cref{thmgivefour} be improved?
    \item For what classes of graphs can we solve the $(k,1)$-cover problem optimally?
\end{itemize}

\appendix

\bibliographystyle{plainurlnat}
\bibliography{main.bib}
\end{document}